\documentclass[lettersize,journal]{IEEEtran}
\usepackage[dvipsnames]{xcolor}
\usepackage{subcaption}
\usepackage{adjustbox} 
\usepackage{anyfontsize}
\usepackage{makecell}
\usepackage{setspace}
\IEEEoverridecommandlockouts
\usepackage[numbers,sort&compress]{natbib}
\usepackage{amsmath,amssymb,amsfonts}
\usepackage{algorithmic}
\usepackage{graphicx}
\usepackage{textcomp}
\usepackage{xcolor}
\usepackage{amsthm}
\usepackage{hyperref}
\usepackage{marvosym}
\usepackage{booktabs,tabulary,multirow,overpic,xcolor}
\usepackage[ruled,vlined,linesnumbered]{algorithm2e}
\hypersetup{
    colorlinks=true,
    linkcolor=blue,  
    urlcolor=violet,  
    citecolor=blue,
}
\newcommand{\subfigref}[2]{\hyperref[#1]{\ref*{#1}.#2}}
\newtheorem{Definition}{Definition}

\newtheorem{Theorem}{Theorem}
\newtheorem{Corollary}{Corollary}
\newtheorem{Proof}{Proof}
\newtheorem{assumption}{Assumption}
\newtheorem{remark}{Remark}

\begin{document}

\title{\huge Temporal-Aware GPU Resource Allocation for Distributed LLM Inference via Reinforcement Learning}
 \author{     \IEEEauthorblockN{Chengze Du$^{1}$, Zhiwei Yu$^1$, Heng Xu$^1$, Haojie Wang$^2$, Bo Liu$^1$, Jialong Li$^{1,}$\textsuperscript{\Letter}} \\
     \IEEEauthorblockA{$^1$ Computer Science and Control Engineering, Shenzhen University of Advanced Technology,
 Shenzhen, China} \\
     \IEEEauthorblockA{$^2$ China Mobile Research Institute, Beijing, China} \\
     \IEEEauthorblockA{\Letter: \text{\texttt{lijialong@suat-sz.edu.cn} 
     }}
 }

\markboth{Journal of \LaTeX\ Class Files, 2025}
{Shell \MakeLowercase{\textit{et al.}}: A Sample Article Using IEEEtran.cls for IEEE Journals}
\maketitle
\begin{abstract}
The rapid growth of large language model (LLM) services imposes increasing demands on distributed GPU inference infrastructure. Most existing scheduling systems follow a reactive paradigm, relying solely on the current system state to make decisions, without considering how task demand and resource availability evolve over time. This lack of temporal awareness in reactive approaches leads to inefficient GPU utilization, high task migration overhead, and poor system responsiveness under dynamic workloads. In this work, we identify the fundamental limitations of these instantaneous-state-only scheduling approaches and propose Temporal Optimal Resource scheduling via Two-layer Architecture (TORTA). TORTA introduces a spatiotemporal scheduling framework that captures both long-term workload patterns and short-term execution constraints. It adopts a two-layer design: a macro-level scheduler leverages reinforcement learning and optimal transport to coordinate inter-region task distribution, while a micro-level allocator refines task-to-server assignments within each region to reduce latency and switching costs. Experimental results across multiple network topologies show that TORTA reduces average inference response time by up to 15\%, improves load balance by approximately 4-5\%, and cuts total operational cost by 10-20\% compared to state-of-the-art baseline methods.
\end{abstract}


\section{Introduction}
\IEEEPARstart{T}{he} rise of large language models (LLMs) has influenced people’s lives, such as in text analysis and image generation~\cite{deepseekai2025deepseekv3technicalreport, claude3, openai2023gpt4, geminiteam2025geminifamilyhighlycapable}. Most LLM inference tasks heavily rely on GPU resources. For example, models like those developed by OpenAI and Anthropic (Claude) require substantial GPU computing power for training and inference, consuming significant energy and computational resources to process complex tasks efficiently~\cite{openai2024pretraining}. Due to geographical, political, and other factors (e.g., economic disparities and infrastructure limitations), the distribution of GPU resources is often uneven compared to user demand~\cite{openai2025status, techcrunch2025gpt45delay}. As shown in Figure~\ref{pic:intro}, global GPU resource distribution is polarized, with most resources concentrated in a few countries and regions (the numbers in the figure represent the number of server clusters)~\cite{computeNorth}. This imbalance significantly affects the quality of LLM product experiences. For instance, some LLM providers frequently fail to deliver timely services to users due to insufficient resources. This issue requires scheduling algorithms to better match supply and demand, accelerate the completion of GPU inference tasks, and thereby enhance user satisfaction with LLM service providers.

 Accelerating GPU inference can be divided into two categories. The first is local acceleration, which involves optimizing the GPU’s operational mode. This includes techniques such as optimizing checkpoint loading, introducing real-time migration and intelligent scheduling strategies, or employing model parallelism and statistical multiplexing to reduce unnecessary overhead~\cite{AlpaServe, ServerlessLLM, Lina, zhong2024distserve}. The second category is distributed resource scheduling, which aims to improve GPU resource utilization and prevent GPU server idleness~\cite{MERL_Yang_2024}. Some approaches focus on allocating suitable GPU resources to different tasks to achieve the fastest inference, often with additional optimization objectives such as minimizing communication costs and power consumption~\cite{Vakilinia, jiang2025demystifying}.

 \begin{figure}[tb]
	\centering
	\includegraphics[width=1.0\linewidth]{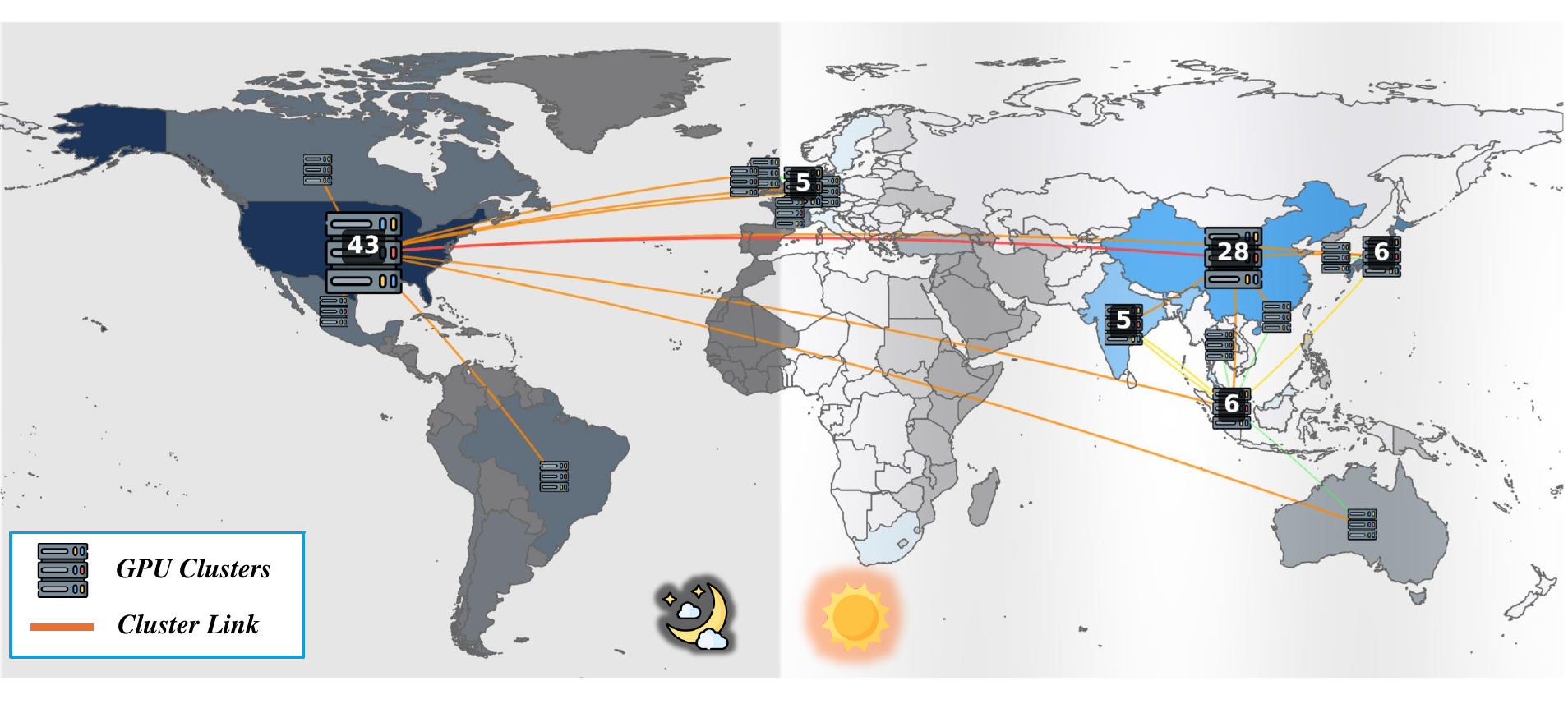}
         \vspace{-0.5cm}
         \caption{Temporal-Aware GPU resource allocation.}
	\label{pic:intro}
\end{figure}

The complexity of this scheduling problem is often very high. Finding an optimal scheduling solution typically requires significant computational time and increases dramatically with network scale expansion, severely impacting task completion time. Existing approaches usually seek approximately optimal solutions within acceptable time constraints, and these methods have achieved good results to a certain extent. For example, MLServerless~\cite{ServerlessLLM} employs an offline-plus-online strategy, where the offline phase invests substantial effort in learning solution strategies, enabling the system to quickly derive good solutions during the online phase to achieve optimal inference speed. Other approaches \cite{jiang2025demystifying, Vakilinia} model this as a Mixed Integer Linear Programming (MILP) or Integer Quadratic Programming (IQP) problem and accelerate solution speed through solution space pruning. However, most of these approaches are \textbf{"Reactive"}, making immediate decisions based on currently observed states (such as cluster interference levels, resource utilization rates, and energy prices), essentially finding optimal solutions within individual time slots. This paradigm inherently ignores the temporal perspective and continuity considerations in decision-making.

Beyond the geographical distribution considerations addressed by previous approaches, we also account for the temporal dependencies in user GPU requests. The temporal distribution provides us with more optimization opportunities. On a smaller time scale, we can more effectively utilize GPU warm-up effects (caching effects)~\cite{lai-warmup} by allocating similar tasks to models whenever possible, thereby reducing inference time. On a larger time scale, we can focus on long-term server utilization and continuous switching costs across clusters, using optimization methods to achieve better load balancing and avoid expensive switching costs~\cite{swap-jens, kun-switch}. However, a critical consideration is that incorporating the temporal dimension increases the solution space by an order of magnitude compared to previously studied problems, undoubtedly increasing the difficulty of solving using traditional methods (such as nested dynamic programming).

Given these challenges, our approach addresses GPU resource allocation by integrating reinforcement learning with an optimal transport~\cite{villani2008optimal} perspective, particularly incorporating the temporal dimension. Unlike traditional methods, our strategy—termed \textbf{T}emporal \textbf{O}ptimal \textbf{R}esource scheduling via \textbf{T}wo-layer \textbf{A}rchitecture (TORTA)—divides the allocation process into two hierarchical layers: macro and micro. In the first layer \texttt{(macro)}, TORTA assigns GPU inference tasks to nodes in each region, primarily to mitigate the imbalance between GPU resources and task distribution. By leveraging optimal transport, inference costs are significantly reduced. In the second layer \texttt{(micro)}, within regional server clusters, TORTA assigns specific GPU inference tasks to suitable servers. The quality of server selection not only reduces task completion time but also avoids oscillations and high overhead caused by frequent server state transitions. Through offline training with large-scale samples, the model can quickly solve problems during the online phase.

Specifically, we model the GPU resource scheduling problem as a Markov Decision Process (MDP) with states encompassing GPU utilization, queue lengths, network latency, and temporal features. The action space involves regional allocation decisions, while the reward function balances response time, load variance, and operational cost. We employ Proximal Policy Optimization (PPO)~\cite{schulman2017-ppo} with optimal transport decisions as supervised signals, enabling the system to learn multi-time-slot allocation strategies that achieve global optimization rather than myopic single-slot decisions.

The major \textit{\textbf{contributions}} of our work are summarized as follows:

\begin{itemize}
    \item \textit{Analysis.} \textbf{We identify the limitations of existing "Reactive" GPU scheduling approaches.} These include a lack of foresight and disregard for temporal dependencies, which can significantly increase the completion time of GPU inference tasks and the switching costs of GPU servers. Additionally, we establish performance upper bound for these methods and derive explicit conditions under which temporal-aware optimization can provably outperform this bound.
    \item \textit{Scheme.} \textbf{We propose a reinforcement learning scheme combined with optimal transport called TORTA.} It integrates spatial and temporal dimensions to achieve better GPU utilization and load balancing compared to traditional approaches. By employing a two-layer resource allocation strategy, the scheme effectively reduces the complexity of solving problems after incorporating the temporal dimension.
    \item \textit{Evaluation.} \textbf{We demonstrate through extensive experiments across various scenarios.} They include four different network topologies and multiple performance metrics. We find that our approach achieves superior response times (16.39s-19.31s vs. 18.72s-24.39s for baselines) and load balance coefficients (0.74-0.76 vs. 0.68-0.73 for baselines) while maintaining lower operational costs with power cost reductions of 10.7K-14.1K compared to existing methods.
\end{itemize}

We organize the remainder of this paper as follows. In Section~\ref{sec: motivation} and \ref{sec: challenges}, we present our motivation, emphasizing the importance of incorporating the temporal dimension in the GPU scheduling problem. Section~\ref{sec: design} provides a detailed description of our methodology, illustrating how reinforcement learning combined with optimal transport is employed to address GPU resource allocation. In Section~\ref{sec: exp}, we conduct extensive experiments to demonstrate the effectiveness of our approach. Sections~\ref{sec: related_works} and~\ref{sec: conclusion} cover related work, and conclude the paper with conclusions and future directions, respectively. Detailed theoretical proofs for performance guarantees, comprehensive network training procedures, and additional experimental results are provided in the Appendix as supplemental material.

\section{Motivation:\\ The Need for Temporal-Aware Allocation}
\label{sec: motivation}
In this section, we primarily answer two questions. 
\begin{enumerate}
    \item \textit{What are the limitations of reactive approaches that do not consider temporal dynamics? }
    \item \textit{How does introducing temporal awareness mitigate or address these issues?}
    \vspace{-0.4cm}
\end{enumerate}

\subsection{Reactive Scheduling and Its Limitations}
\label{sec: motivation_1}

Reactive scheduling strategies dominate current GPU resource allocation systems. These methods make decisions based solely on the system's state at the current moment, lacking the ability to predict future load patterns. This short-sighted decision-making approach is particularly inefficient when handling predictable traffic patterns.

Figure~\ref{fig:motivation_1} illustrates the performance difference between reactive and predictive scheduling when handling periodic traffic peaks. Reactive approach only respond passively after a traffic surge (as seen in the red portion of Figure~\subfigref{fig:motivation_1}{a}), leading to a significant number of requests waiting in queues. Figure~\subfigref{fig:motivation_1}{b} shows the distribution of GPU inference task queuing times during traffic surges, revealing a clear "bimodal pattern" where queuing times are predominantly either long or short, with long queuing times significantly exceeding short ones. In real-world scenarios, this greatly impacts the user experience of LLM services. In contrast, an ideal predictive approach (the blue dashed line in Figure~\subfigref{fig:motivation_1}{a}) prepares resources 15 minutes in advance through forecasting, resulting in a smoother overall power increase. 

Figure~\subfigref{fig:motivation_1}{c} displays the variation in average queuing time throughout the process, showing a clear "staircase effect". The system only begins scaling up after detecting a load increase, and GPUs require 1-3 minutes to transition from cold start to full readiness. During this period, accumulated requests result in an average queuing time of 15.7 seconds. Queuing time gradually decreases as servers scale up, eventually dropping to within 1 second after the GPU request load stabilizes.

The absence of foresight affects not only immediate service quality but also triggers cascading effects. Sudden queuing can prompt client retries, further increasing system load~\cite{ha2024retrystorm, huang2022metastablefail}. Additionally, passive scaling often overreacts, leading to resource waste after traffic subsides~\cite{Naseer2020ZDR}. These issues highlight the necessity of incorporating temporal awareness mechanisms.
\begin{figure}[t]
	\centering
	\includegraphics[width=1.02\linewidth]{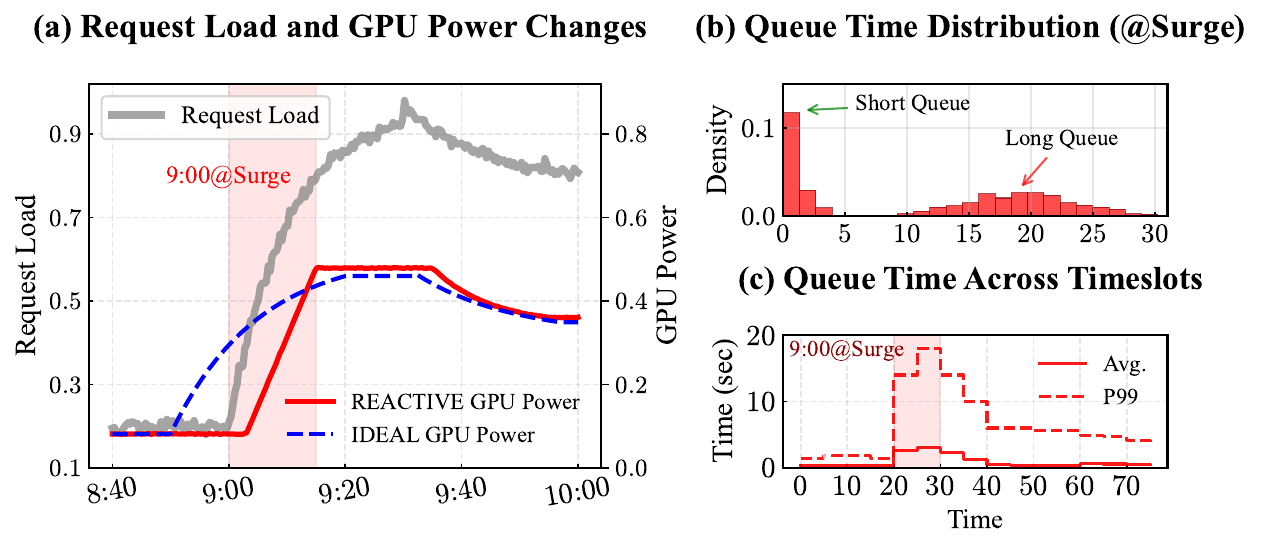}
         \vspace{-0.2cm}
	\caption{Visualization of the limitations of reactive scheduling.}
	\label{fig:motivation_1}
        \vspace{-0.5cm}
\end{figure}

\subsection{Costs of Ignoring Decision Continuity}
\label{sec: motivation_2}

Single-time-slot optimization overlooks the temporal dependencies between decisions. Each scheduling decision is treated as an independent event, making it difficult to account for the costs of state transitions and the lasting impact of decisions. This fragmented optimization approach leads to drastic configuration changes between adjacent time slots, resulting in unnecessary task migrations and resource switching.

Figure~\subfigref{fig:motivation_2}{a} illustrates the hidden costs of task migration and model switching. For a typical large model (LLaMA-2-7B, deployed on an V100), migrating a task from node A to node B involves: (1) model state serialization (\textasciitilde15.2 seconds), (2) target node deserialization (\textasciitilde4.8 seconds) and GPU memory loading (\textasciitilde5.6 seconds), and (3) inference engine warm-up (\textasciitilde5.1 seconds). Switching between different models on the same server (e.g., from LLaMA-2-7B to Qwen-7B) also incurs significant costs, including unloading the current model (\textasciitilde3.5 seconds), memory cleanup (\textasciitilde2.1 seconds), loading the new model (\textasciitilde6.8 seconds), state initialization (\textasciitilde14.2 seconds), and engine reconfiguration (\textasciitilde3.4 seconds). We conducted tests on various GPU types (shown in Figure~\subfigref{fig:motivation_2}{b}). The figure indicates that the V100 exhibits higher migration costs across all stages compared to the H100, RTX 4090 and RTX 3090. These experimental results demonstrate that even task migration and model switching between small-scale language models incur substantial system overhead.

GPU power consumption varies significantly across different stages shown in Figure~\subfigref{fig:motivation_2}{c}. During deserialization and memory loading, power consumption increases notably and exhibits significant fluctuations. For an V100 with a power consumption of 250W, the peak power can reach 237W. Frequent cold-to-hot transitions in GPUs may accelerate hardware aging and adversely affect grid stability~\cite{Zhabelova2015dcpower}. 

These observations reveal a fundamental limitation, reactive methods face a performance upper bound determined by optimal single-timeslot solutions plus unavoidable switching costs. Our analysis in Appendix proves this upper bound exists and provides conditions for surpassing it through temporal-aware optimization. And we aim to design smoother state transition strategies that reduce switching overhead while maintaining flexibility.

\begin{figure}[t]
	\centering
	\includegraphics[width=1.02\linewidth]{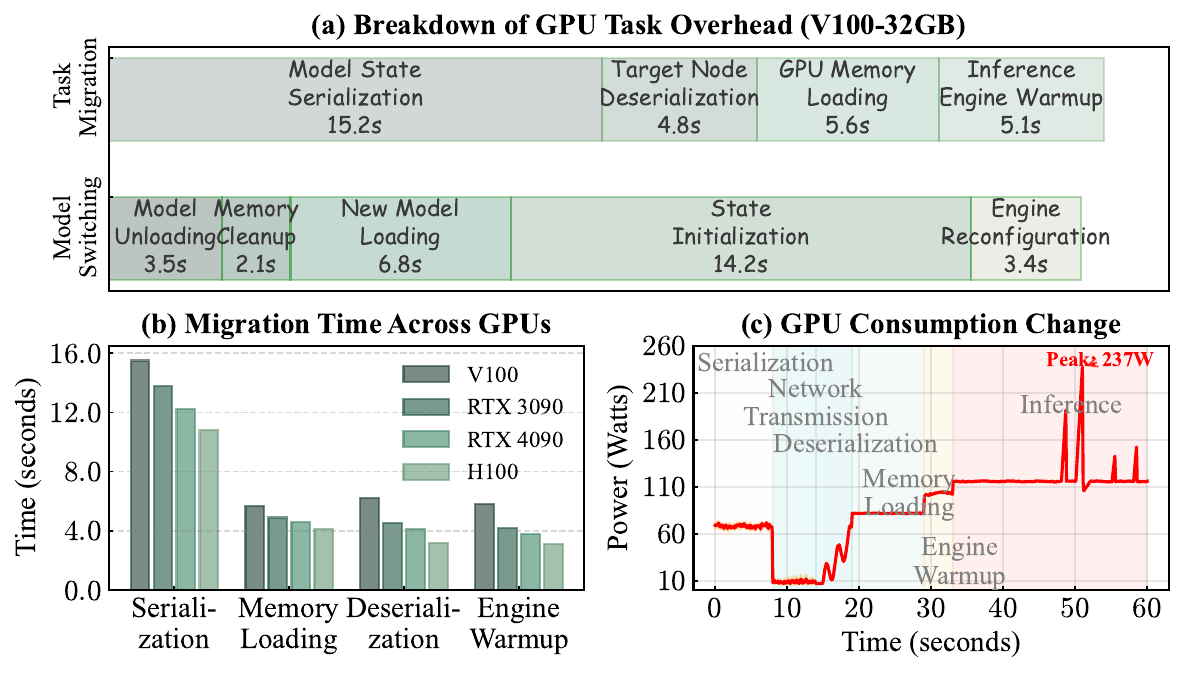}
         \vspace{-0.2cm}
	\caption{Breakdown of GPU Task Overhead and Power Consumption During Task Migration and Model Switching.}
	\label{fig:motivation_2}
        \vspace{-0.3cm}
\end{figure}

\subsection{Allocation Inefficiency Without Coordination}
\label{sec: motivation_3}
Reactive decision-making lacks foresight and coordination across multiple time slots. This short-sightedness becomes particularly evident when addressing sudden events. Figure~\subfigref{fig:motivation_3}{a} illustrates the geographical distribution of GPU resources and their inter-regional connections, highlighting the limitations of reactive strategies when a critical region experiences a failure (as indicated by “CRITICAL FAILURE” in the figure) which we stimulated in our system environment.

Figure~\subfigref{fig:motivation_3}{c} visualizes the recovery process of traditional single-time-slot methods in response to sudden GPU failures. Upon detecting a failure, these methods, lacking foresight into future load and resource conditions, tend to blindly migrate affected tasks to the nearest available regions within the first time slot (T1). This immediate but shortsighted decision causes a sharp increase in task queuing latency in these neighboring regions during T1, with a significant portion of tasks being dropped due to resource overload. Recovery occurs gradually over the subsequent three time slots (T2, T3, T4), but this “delayed response” significantly impacts user experience and system efficiency.

\begin{figure}[t]
	\centering
	\includegraphics[width=1.02\linewidth]{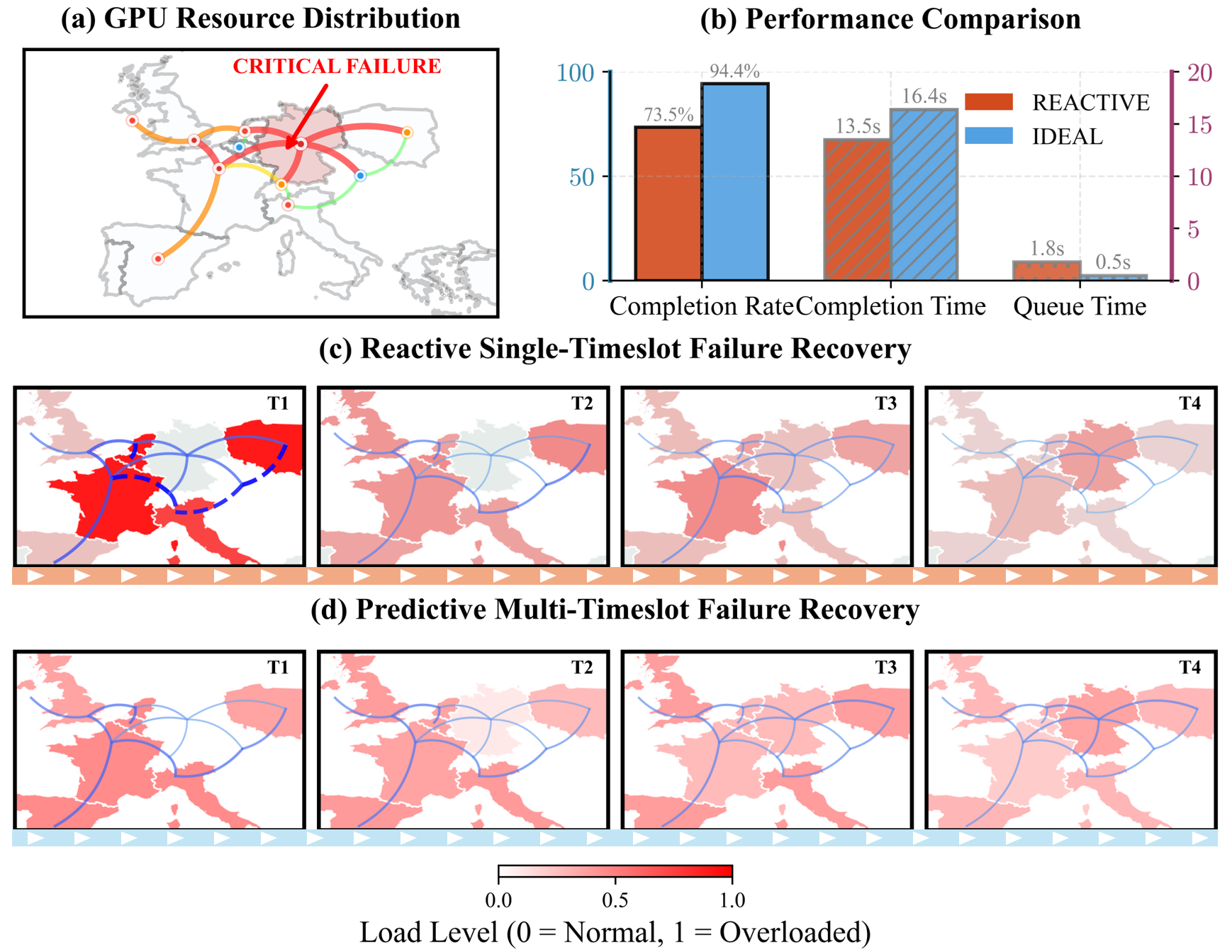}
         \vspace{-0.4cm}
	\caption{Comparison of Reactive and Predictive Scheduling Strategies Inference Under Critical Failure Conditions.}
	\label{fig:motivation_3}
        \vspace{-0.5cm}
\end{figure}

Figure~\subfigref{fig:motivation_3}{d} demonstrates the recovery strategy of multi-time-slot predictive methods under the same failure scenario. By integrating temporal information and predictive capabilities, this approach intelligently and evenly distributes tasks across multiple future time slots and geographical regions, effectively avoiding local overload. This smooth migration strategy not only significantly improves overall task completion rates and reduces queuing times but also shortens average completion times. As shown in Figure~\subfigref{fig:motivation_3}{b}, the reactive methods have shorter completion time, primarily due to its task-dropping mechanism that seriously affects user experience. The predictive approach exhibits clear advantages in performance metrics (completion rate, and queuing time) during failure recovery.

\section{Key Challenges}
\label{sec: challenges}

Although introducing temporal-aware mechanisms appears to address some limitations of reactive methods discussed in \S~\ref{sec: motivation}, their implementation still faces several challenges. 

\vspace{-0.3cm}
\subsection{Multi-Objective Optimization Complexity}
\label{sec: challenge_1}
Optimizing GPU resource allocation requires balancing response time, operational costs, and load balancing across distributed clusters, posing a complex combinatorial optimization problem. When $N$ tasks arrive in each time window and $M$ regional clusters each contain $K$ servers, even single-time-slot allocation yields  $O((M \times K)^N)$ possible solutions. Real-world scenarios often involve over 50 clusters and more than 100,000 tasks, making exhaustive search infeasible. Heterogeneous task types (e.g., compute-intensive or memory-intensive) require distinct resource-matching strategies, further increasing problem complexity.

Current methods model this as a mixed-integer linear programming (MILP) problem, but the computational cost becomes impractical in large-scale scenarios. As shown in Figure~\subfigref{fig: challenge}{a}, MILP solving time grows exponentially with task volume—handling 5,000 tasks takes over 2 minutes on our experimental hardware\footnote{Tested on an Intel i5-13490F (10 cores, 4.8GHz boost).}, while production systems require sub-second decisions~\cite{zhong2024distserve}. Note that our experimental scenario represents a significantly simplified version of real-world deployments, which involve dozens of server types with heterogeneous GPU architectures, diverse task categories with complex resource requirements, dynamic pricing models, and multi-tenant constraints. The computational complexity grows exponentially with these additional dimensions, making traditional MILP approaches fundamentally unsuitable for real-time multi-objective optimization. 

\begin{figure}[t]
    \fontsize{6pt}{1pt}\selectfont
    \centering
    \begin{minipage}[c]{0.47\linewidth}
        \includegraphics[width=\linewidth]{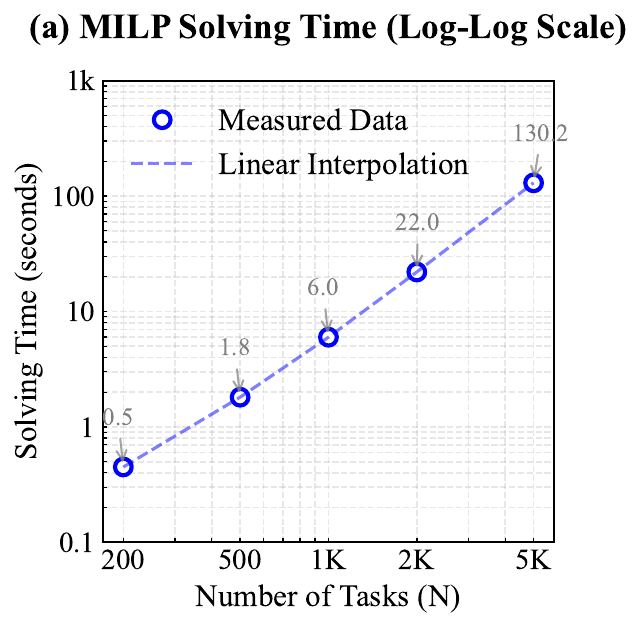}
    \end{minipage}
    \hfill
    \begin{minipage}[c]{0.50\linewidth}
        \renewcommand{\arraystretch}{0.1} 
        {\scriptsize (b) Configuration parameters \centering \par}
        \begin{tabular}{@{}p{0.23\linewidth}p{0.33\linewidth}p{0.28\linewidth}@{}}
        \toprule
        \textbf{Component} & \textbf{Configuration} & \textbf{Scale} \\
        \midrule
        Infrastructure & 5 regions, 10 servers each & 50 total servers \\
        \midrule
        Task Properties & 2 types (compute/memory) & Variable arrival rates \\
        \midrule
        Capacity Model & Dynamic Server limits & 3-20 tasks per server \\
        \midrule
        Load Balancing & Max 80\%  Per region & Prevents concentration \\
        \midrule
        Variables & Binary matrix & $N \times 50$ decisions \\
        \midrule
        Constraints & Assignment + capacity & $N + 55$ constraints \\ 
        \bottomrule
        \end{tabular}
    \end{minipage}
    \caption{MILP strategy solving time.}
    \vspace{-0.6cm}
    \label{fig: challenge}
\end{figure}

\vspace{-0.3cm}
\subsection{How to Optimize for Long-term Benefits?}
\label{sec: challenge_2}
Optimizing GPU resources across multiple time slots requires reasoning about future states and their impacts. The system must learn complex temporal patterns: how current allocation decisions affect future load distribution, when to preheat servers for anticipated demand, and how to maintain stability amid changes. Each decision triggers cascading effects—allocating tasks to a specific region not only impacts immediate load but also alters migration patterns, queue accumulation, and resource availability in subsequent time slots. These temporal dependencies fundamentally create a long-term credit assignment challenge, where:

The core difficulty manifests in several ways. First, decisions with delayed consequences—like keeping servers preheated—may initially appear suboptimal but prove critical during demand surges. Conversely, aggressive task packing for immediate efficiency often induces cascading migrations that degrade overall performance. Second, the system must infer these long-term causal relationships from historical data while remaining adaptive to new patterns. Most critically, the high-dimensional state space (encompassing load distribution, server states, and queue lengths across regions) combined with a continuous action space renders traditional dynamic programming methods infeasible.

\section{Torta Overview}
We propose a scheme called \textbf{T}emporal \textbf{O}ptimal \textbf{R}esource scheduling via \textbf{T}wo-layer \textbf{A}rchitecture \textbf{(TORTA)} to address the limitations of reactive scheduling in \S\ref{sec: motivation} and the challenges of multi-objective optimization complexity in \S\ref{sec: challenges}. This scheme systematically addresses the shortcomings of existing methods in handling temporal dependencies, decision continuity, and long-term optimization through temporal dimension modeling and a two-layer optimization architecture.

We remodel GPU resource allocation as a spatio-temporal coupled optimization problem. This considers both spatial heterogeneity (GPU resource distribution and performance differences across geographical regions) and temporal dynamics (load fluctuations, state transition costs, and long-term dependencies). Through this unified spatiotemporal modeling, the system can fully consider the impact of current allocation decisions on future system states when making decisions, thereby avoiding the "short-sighted" problems caused by traditional reactive methods. The optimization objectives include not only immediate task completion time and computational costs, but also minimizing state transition overhead across time slots.

For this optimization problem, we decomposes the originally exponential-complexity global optimization problem into two relatively independent but coordinated subproblems by designing a two-layer hierarchical architecture:

$\circ$ \textit{Macro Regional Allocation (\S\ref{subsec: marco}):} This layer handles strategic resource scheduling tasks across regions, focusing on solving geographical distribution imbalances and large-scale load balancing problems. The layer first uses a \textbf{demand predictor} (\S\ref{subsubsec: rl}) to analyze historical load patterns and temporal features, predicting request distributions for future time slots in each region to provide decision basis for subsequent optimization. The core \textbf{RL-OT optimizer} combines the long-term planning capability of reinforcement learning with the supply-demand matching advantages of optimal transport theory to generate inter-regional allocation strategies that consider temporal continuity.

$\circ$ \textit{Micro Server Selection (\S\ref{subsec: micro}):} Under the regional allocation constraints determined by the macro layer, this layer focuses on fine-grained task-server matching problems within clusters. The task-server optimization engine comprehensively evaluates hardware compatibility, real-time load status, and server warm-up conditions to compute optimal server allocation schemes for each task, and continuously optimizes initial allocation results through local search algorithms. The state manager coordinates the timing of server state transitions, intelligently deciding the switching timing of each server between different operating states based on load predictions, significantly reducing the high switching costs analyzed in \S\ref{fig:motivation_2} through proactive state management.

\section{Design of TORTA}
\label{sec: design}
We outline our problem modeling and optimization objectives and then divide our scheme into macro and micro components for detailed explanation.
\subsection{Problem Formulation}
\label{subsec: prob_formulation}

\begin{figure}[t]
	\centering
	\includegraphics[width=\linewidth]{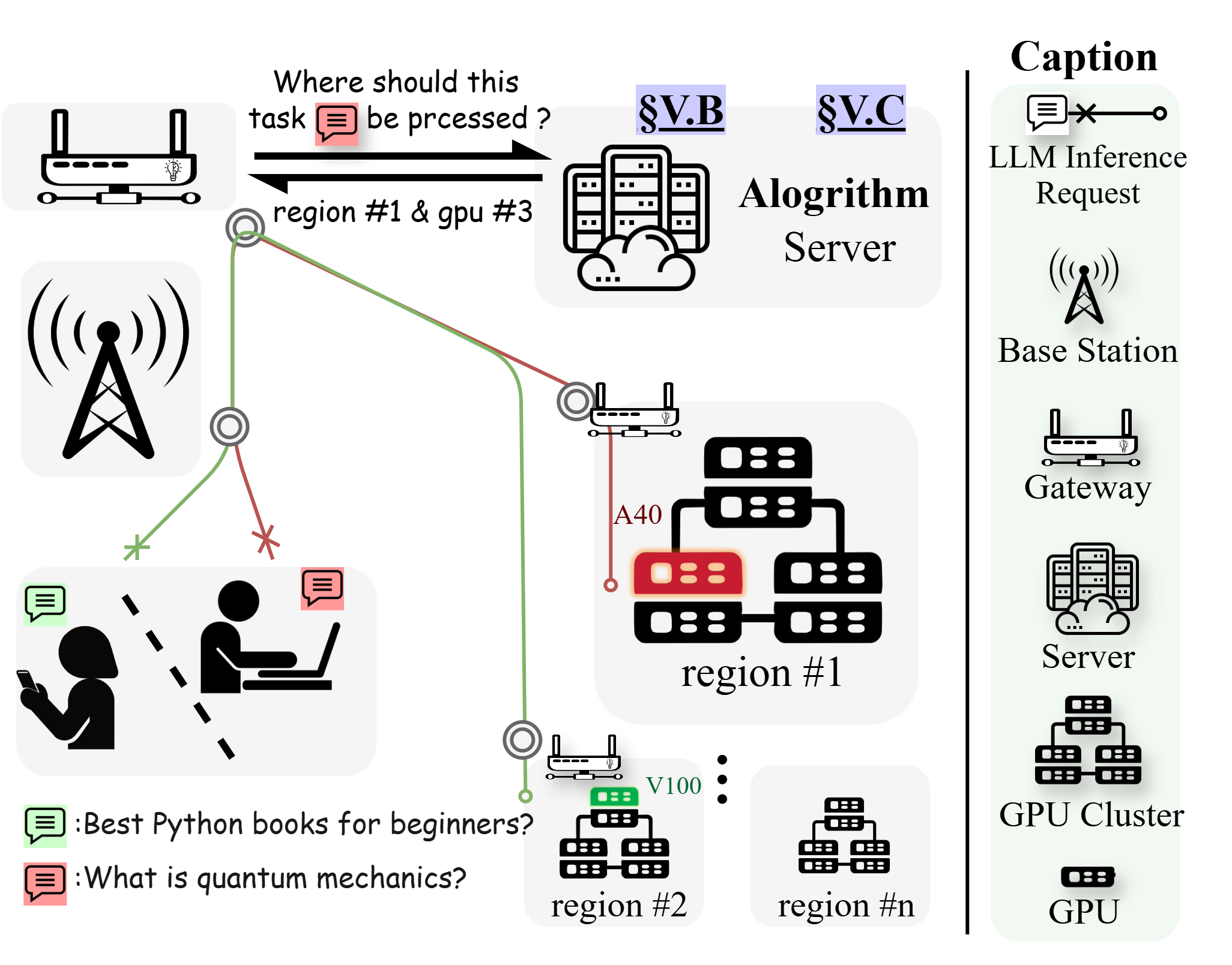}
         \vspace{-0.4cm}
	\caption{Overview of TORTA's system architecture.}
	\label{fig: overview}
        \vspace{-0.5cm}
\end{figure}

We formalize the temporal-aware GPU resource allocation problem as a spatiotemporal optimization problem with multiple time slots. Consider a distributed system consisting of $R$ geographical regions, where each region $r \in \{1, 2, \ldots, R\}$ contains $S_r$ GPU servers. The optimization is performed over a time horizon of $T$ discrete time slots, where each time slot has a duration of $\Delta t$.

Let $\mathcal{N}_t = \{n_1^t, n_2^t, \ldots, n_{N_t}^t\}$ denote the set of GPU inference tasks arriving at time slot $t$, where $N_t$ represents the total number of tasks in slot $t$. Each task $n_i^t$ is characterized by a tuple $(c_i^t, m_i^t, d_i^t)$, representing its computational requirements, memory requirements, and deadline constraint, respectively. The system state at time slot $t$ is defined as $s_t = (U_t, Q_t, L_t, H_t)$, where $U_t \in \mathbb{R}^{R \times S_{max}}$ represents the utilization matrix of all servers across regions, $Q_t \in \mathbb{R}^R$ denotes the queue lengths in each region, $L_t \in \mathbb{R}^{R \times R}$ is the inter-region latency matrix, and $H_t \in \mathbb{R}^{R \times S_{max}}$ captures the historical load patterns and temporal features.

The allocation decision at time slot $t$ is represented by the assignment matrix $X_t \in \{0, 1\}^{N_t \times R \times S_{max}}$, where $X_{i,r,s}^t = 1$ indicates that task $n_i^t$ is assigned to server $s$ in region $r$, and 0 otherwise. The allocation must satisfy the following constraints: each task is assigned to exactly one server, i.e., $\sum_{r=1}^R \sum_{s=1}^{S_r} X_{i,r,s}^t = 1$ for all $i \in \{1, \ldots, N_t\}$. 

The \textbf{objective function} consists of three components: response time, power cost, and switching cost.

\underline{The response time} for task $n_i^t$ assigned to server $s$ in region $r$ is modeled as $T_{completion}(i,r,s,t) = T_{queue}(r,s,t) + T_{compute}(i,r,s) + T_{network}(i,r)$\footnote{Network latency is typically 1\% of compute time. For reference, GPT-4 generates tokens at 13 tokens/s; a 250-word conversation takes $\sim$25s, while intercontinental latency ranges between 10--100ms.}, where $T_{queue}(r,s,t)$ represents the queuing delay, $T_{compute}(i,r,s)$ is the actual computation time, and $T_{network}(i,r)$ accounts for network transmission delays.

\underline{The power cost} $C_{power}^t$ captures the electricity consumption costs across geographical regions, which vary significantly due to differences in regional electricity pricing and infrastructure efficiency. This cost component accounts for both the computational power required for task processing and the regional cost variations that enable cost-aware geographical optimization. Formally, $C_{power}^t = \sum_{i=1}^{N_t} \sum_{r=1}^R \sum_{s=1}^{S_r} X_{i,r,s}^t \cdot (w_{network} \cdot L_{origin(i),r} + w_{compute} \cdot P_{r,s})$, where $origin(i)$ denotes the origin region of task $i$, $w_{network}$ and $w_{compute}$ are weight parameters, $L_{origin(i),r}$ represents the network transmission cost, and $P_{r,s}$ represents the computational power cost at server $s$ in region $r$.

\underline{The switching cost} $C_{switch}^t$ directly measures the allocation differences between consecutive timeslots, capturing the overhead of changing resource allocation patterns over time. This cost is crucial for capturing the temporal dependencies identified in our motivation analysis. We define $C_{switch}^t = ||X_t - X_{t-1}||_F^2$, where $X_t$ and $X_{t-1}$ represent the allocation matrices at timeslots $t$ and $t-1$ respectively, and $||\cdot||_F$ denotes the Frobenius norm. This formulation directly penalizes dramatic allocation changes between consecutive timeslots, encouraging temporal smoothness while allowing necessary adaptations to changing demand patterns.

The complete optimization problem can be formulated as finding an allocation policy $\pi: \mathcal{S} \rightarrow \mathcal{A}$ that minimizes the expected cumulative cost over the entire time horizon:
\begin{equation}
\begin{aligned}
\min_{\pi} \mathbb{E} \Bigg[ & \sum_{t=1}^{T} \bigg( \sum_{i=1}^{N_t} T_{completion}(i,r,s,t) \\
& \quad + \alpha \cdot C_{switch}^t + \beta \cdot C_{power}^t \bigg) \Bigg]
\end{aligned}
\end{equation}

where $\alpha$ and $\beta$ are weight parameters that serve dual purposes: (1) they balance the relative importance of switching costs and power costs against response time, and (2) they normalize these heterogeneous metrics (time in seconds, power in watts, switching overhead in computational cycles) into a unified optimization objective with consistent units. 

\subsection{Macro-Level Allocation: OT with RL}
\label{subsec: marco}

The macro-level allocation layer addresses the strategic distribution of GPU inference tasks across geographical regions. This layer operates by first establishing an upper-bound allocation strategy using optimal transport theory (as proved in Appendix), which then serves as a supervised signal for reinforcement learning optimization. The two-stage scheme enables the system to leverage the theoretical optimality guarantees of optimal transport while adapting to complex temporal dynamics through reinforcement learning.
\begin{figure}[t]
	\centering
	\includegraphics[width=0.98\linewidth]{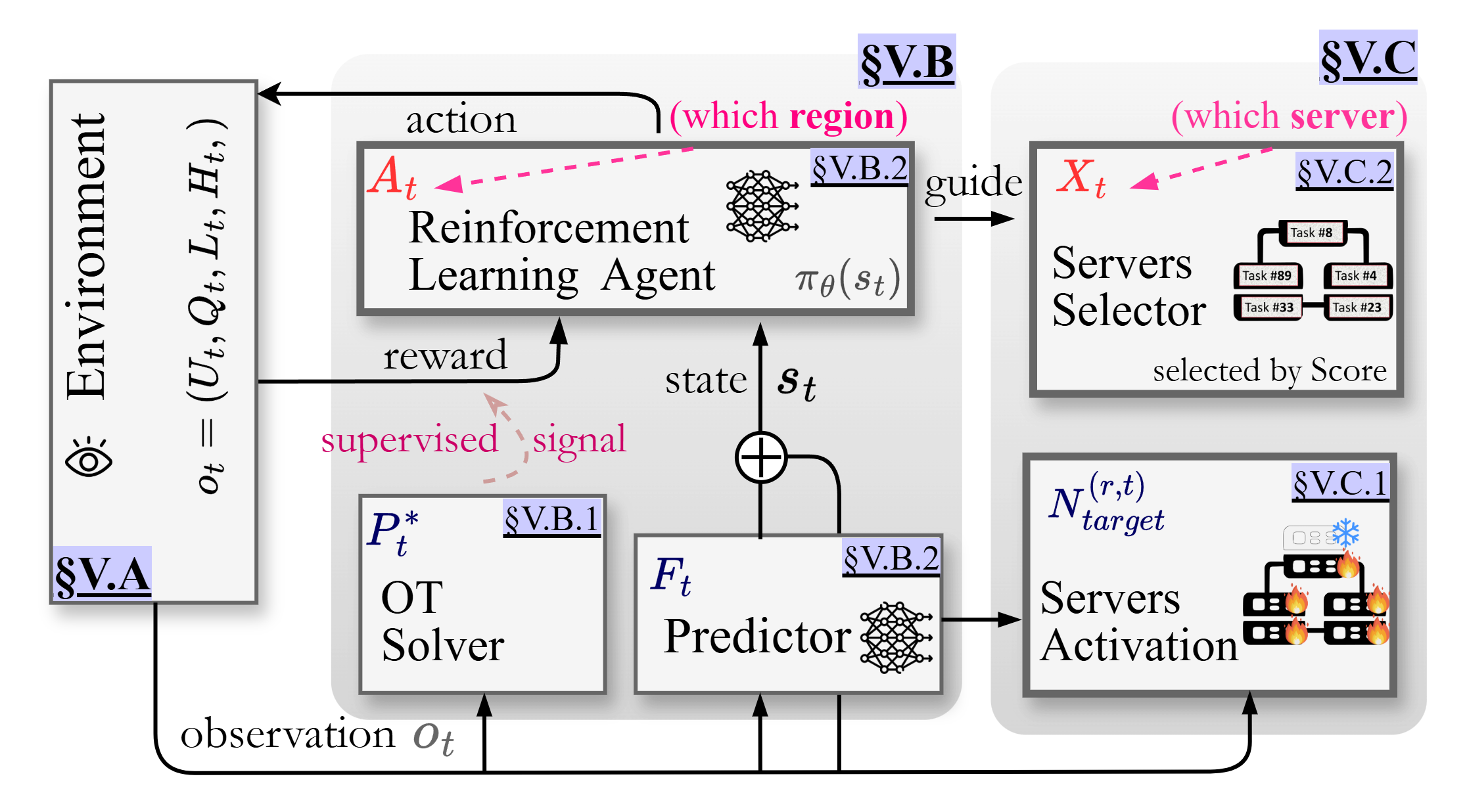}
	\caption{Key component of TORTA.}
	\label{fig: component}
        \vspace{-0.5cm}
\end{figure}

\subsubsection{Optimal Transport for Regional Load Balancing}

We model the regional allocation problem as an optimal transport problem between GPU request distributions and GPU resource distributions. Let $\mu_t \in \mathbb{R}^R$ represent the normalized GPU request distribution at time slot $t$, where $\mu_t^{(r)}$ denotes the proportion of total requests originating from region $r$. Similarly, let $\nu_t \in \mathbb{R}^R$ represent the normalized GPU resource distribution, where $\nu_t^{(r)}$ indicates the proportion of total GPU capacity available in region $r$.

In practice, the sum of raw request counts and resource capacities across regions are typically unequal, i.e., $\sum_{r=1}^R RequestCount_r^t \neq \sum_{r=1}^R ResourceCapacity_r^t$. To apply optimal transport theory, we normalize both distributions to unit mass: $\mu_t^{(r)} = \frac{RequestCount_r^t}{\sum_{i=1}^R RequestCount_i^t}$ and $\nu_t^{(r)} = \frac{ResourceCapacity_r^t}{\sum_{i=1}^R ResourceCapacity_i^t}$, ensuring $\sum_{r=1}^R \mu_t^{(r)} = \sum_{r=1}^R \nu_t^{(r)} = 1$.

The cost matrix $C \in \mathbb{R}^{R \times R}$ captures the cost between regions, where $C_{i,j}$ represents the cost of serving a request from region $i$ using resources in region $j$. This cost primarily incorporates regional power cost differences and minimal network overhead: $C_{i,j} = w_1 \cdot (PowerCost_j) + w_2 \cdot (L_{i,j} + BandwidthCost_{i,j})$, where $PowerCost_j$ represents the electricity cost in region $j$, $L_{i,j}$ is the network latency between regions $i$ and $j$, $BandwidthCost_{i,j}$ accounts for data transfer overhead, and $w_1 \gg w_2$ reflects the dominance of power costs over network costs in the optimization objective.

The optimal transport problem seeks to find the allocation matrix $P^* \in \mathbb{R}^{R \times R}$ that minimizes the total transportation cost while satisfying the marginal constraints:
\begin{equation}
    P^* = \arg\min_{P \geq 0} \sum_{i=1}^R \sum_{j=1}^R C_{i,j} P_{i,j}
\end{equation}
subject to the constraints $\sum_{j=1}^R P_{i,j} = \mu_t^{(i)}$ for all $i$ (demand satisfaction) and $\sum_{i=1}^R P_{i,j} = \nu_t^{(j)}$ for all $j$ (capacity constraints).

The solution $P^*$ provides an upper-bound allocation strategy where $P^*_{i,j}$ represents the optimal mass flow from region $i$ to region $j$. To convert this into actionable routing probabilities, we normalize each row of the allocation matrix: $Prob_{i \rightarrow j} = \frac{P^*_{i,j}}{\sum_{k=1}^R P^*_{i,k}}$. This normalization yields a row-stochastic matrix where $Prob_{i \rightarrow j}$ represents the probability that a task originating from region $i$ should be routed to region $j$ for processing.

The resulting probability distribution $Prob_i = [Prob_{i \rightarrow 1}, Prob_{i \rightarrow 2}, \ldots, Prob_{i \rightarrow R}]$ for each source region $i$ serves as a reference allocation strategy that balances supply and demand optimally under the given cost structure. However, this static optimal transport solution does not account for temporal dependencies, future load predictions, or the dynamic nature of system states. To address these limitations, we integrate the optimal transport solution as a supervised signal within our reinforcement learning framework.

\subsubsection{Reinforcement Learning with OT Supervision}
\label{subsubsec: rl}
The reinforcement learning component learns to generate temporally smooth allocation matrices that respect the optimal transport baseline while minimizing abrupt load transitions between consecutive time slots. The primary objective is to produce allocation decisions that maintain system stability and avoid the cascading effects of sudden load redistribution across regions.

We formulate the problem as a Markov Decision Process (MDP) with state space $\mathcal{S}$, action space $\mathcal{A}$, transition function $\mathcal{T}$, and reward function $\mathcal{R}$. The state at time slot $t$ is defined as $s_t = (U_t, Q_t, L_t, H_t, F_t, A_{t-1})$, where $U_t \in \mathbb{R}^{R}$ represents the current utilization levels across regions, $Q_t \in \mathbb{R}^{R}$ denotes the queue lengths in each region, $L_t \in \mathbb{R}^{R \times R}$ is the inter-region latency matrix, $H_t$ captures historical load patterns and temporal features, $F_t \in \mathbb{R}^{R}$ is the predicted task arrival for the next time slot obtained from a pre-trained demand predictor, and $A_{t-1} \in [0,1]^{R \times R}$ is the allocation matrix from the previous time slot.

The demand predictor is trained offline using historical data to forecast future task arrivals. Given historical features $\{U_{t-k}, Q_{t-k}, H_{t-k}\}_{k=1}^K$ for the past $K$ time slots, the predictor estimates $F_t = \text{Predictor}(U_{t-K:t}, Q_{t-K:t}, H_{t-K:t})$, where $F_t^{(r)}$ represents the predicted number of tasks arriving in region $r$ during time slot $t+1$. And the demand predictor is implemented as a neural network trained offline using historical data.

The action space consists of allocation matrices $A_t \in [0,1]^{R \times R}$ where $A_t^{(i,j)}$ represents the proportion of tasks from region $i$ that should be allocated to region $j$. The action must satisfy the normalization constraint $\sum_{j=1}^R A_t^{(i,j)} = 1$ for all source regions $i \in \{1, 2, \ldots, R\}$, ensuring that all tasks from each region are allocated.

The reward function is designed to balance three objectives: alignment with the optimal transport solution, temporal smoothness, and load stability. The reward at time slot $t$ is formulated as:
\begin{equation}
\begin{aligned}
r_t = r_{OT}(A_t, P^*_t) + \lambda_1 \cdot r_{smooth}(A_t, A_{t-1}) + \lambda_2  \cdot r_{cost}(Q_t)
\end{aligned}
\end{equation}
where $r_{OT}(A_t, P^*_t) = -\|A_t - P^*_t\|_F^2$ measures the alignment between the RL allocation matrix and the optimal transport solution, encouraging the agent to respect the theoretical optimality while allowing adaptive deviations.

The smoothness term $r_{smooth}(A_t, A_{t-1}) = -\|A_t - A_{t-1}\|_F^2$ penalizes large changes between consecutive allocation matrices, promoting temporal continuity and preventing abrupt load redistributions that can cause system instability. The cost term $r_{\text{cost}}(Q_t) = -\frac{\|Q_t\|_1}{Q_{\text{max}}}$ utilizes the aggregate backlog across regional servers as a real-time proxy for system responsiveness, where $Q_t$ represents the vector of pending task counts and $Q_{\text{max}}$ denotes the maximum tolerable queue capacity. The weights $\lambda_1$ and $\lambda_2$ are set to prioritize OT alignment while ensuring sufficient temporal smoothness and load balance, with the relative weighting empirically tuned to achieve stable convergence during training.

The policy network $\pi_\theta: \mathcal{S} \rightarrow \mathcal{A}$ is parameterized by $\theta$ and trained using Proximal Policy Optimization (PPO). The network outputs the parameters of a Beta distribution for each element of the allocation matrix, ensuring that the generated probabilities lie within $[0,1]$ and can be appropriately normalized. The policy gradient objective incorporates the temporal smoothness directly into the learning process:
\begin{equation}
\begin{aligned}
L(\theta) = \mathbb{E}_{s_t, a_t \sim \pi_\theta} \Bigg[ 
    & \min\bigg( \frac{\pi_\theta(a_t|s_t)}{\pi_{\theta_{old}}(a_t|s_t)} \hat{A}_t, \\
    & \text{clip}\left(\frac{\pi_\theta(a_t|s_t)}{\pi_{\theta_{old}}(a_t|s_t)}, 1-\epsilon, 1+\epsilon\right) \hat{A}_t \bigg) \Bigg]
\end{aligned}
\end{equation}
where $\hat{A}_t$ is the advantage estimate computed using the temporal reward structure, and $\epsilon$ is the clipping parameter. The advantage estimation incorporates the smoothness penalty to encourage policies that maintain allocation stability across time slots.

To achieve provable performance superiority over reactive approaches (which focus on single-timeslot optimization) in multi-timeslot joint optimization scenarios, we enhance the standard PPO training objective with carefully designed constraint terms. Our modified loss function is formulated as:
\begin{equation}
\begin{aligned}
\mathcal{L}_{\text{total}} = \mathcal{L}_{\text{PPO}} + \gamma \cdot \mathcal{L}_{\epsilon}(s, \epsilon) + \delta \cdot \mathcal{L}_{s}(s, \epsilon)
\end{aligned}
\end{equation}
where $\mathcal{L}_{\epsilon}(s, \epsilon) = \max(0, \frac{\|A_t^{\text{RL}} - A_t^{\text{OT}}\|_F - \epsilon_{\max}}{\epsilon_0})$ controls the maximum deviation from the optimal transport baseline, and $\mathcal{L}_{s}(s, \epsilon) = \max(0, \frac{s_{\min} - s_{\text{current}}}{s_0})$ ensures sufficient improvement in switching cost reduction.

This design ensures that TORTA operates within the theoretical trade-off region where reduced switching costs outweigh single-timeslot optimality losses. The constraint parameters $\epsilon$ and $s$ must satisfy a specific relationship to guarantee theoretical performance advantages over any reactive method: $\frac{1-1/s}{\epsilon} > \frac{L_R + \beta L_P}{K_0}$, where $L_R$, $L_P$ are Lipschitz constants of response time and power cost functions, and $K_0$ is the baseline switching cost of reactive methods. The rigorous mathematical derivation of this relationship, detailed training procedures, and the formal proof are provided in Appendix.

\subsection{Micro-Level Allocation}
\label{subsec: micro}
The micro-level allocation layer operates within each regional cluster to assign specific GPU inference tasks to appropriate servers. This layer receives the regional allocation decisions from the macro layer and focuses on optimizing task-server matching based on hardware compatibility, current load conditions, and anticipated future demands. The micro-level optimization addresses two key decisions: determining the optimal number of active servers and selecting the most suitable server for each incoming task.

\subsubsection{Dynamic Server Activation Strategy}

The number of active servers in each region is dynamically adjusted based on current workload and predicted future demands. Let $N_{active}^{(r,t)}$ denote the number of active servers in region $r$ at time slot $t$, where $1 \leq N_{active}^{(r,t)} \leq S_r$. The server activation decision balances immediate processing capacity with operational costs and transition overheads.

The activation strategy considers three factors: current queue length $Q_t^{(r)}$, predicted incoming tasks $F_t^{(r)}$ from the demand predictor, and the warm-up costs associated with transitioning servers from idle to active states. The target number of active servers is computed as:

\begin{equation}
\label{eq: activate}
\begin{aligned}
    N_{target}^{(r,t)} = \min\left( S_r, \left\lceil \frac{Q_t^{(r)} + F_t^{(r)} + \sigma \cdot \sqrt{F_t^{(r)}}}{C_{avg}^{(r)}} \right\rceil \right)
\end{aligned}
\end{equation}
where $C_{avg}^{(r)}$ represents the average processing capacity per server in region $r$, and $\sigma$ is a safety factor that accounts for load variability. The square root term provides a buffer proportional to the standard deviation of the predicted load, following queueing theory principles.

To minimize transition costs, servers are activated or deactivated gradually. If $N_{target}^{(r,t)} > N_{active}^{(r,t-1)}$, additional servers are selected for activation based on their warm-up time and hardware compatibility with incoming tasks. Conversely, when $N_{target}^{(r,t)} < N_{active}^{(r,t-1)}$, servers with the lowest current utilization and longest idle periods are prioritized for deactivation to minimize wasted warm-up investments.

\subsubsection{Greedy Task-Server Matching}

For the active servers, we employ a greedy algorithm that iteratively assigns each task to the server with the highest compatibility score. The score {\small$ Score(task_i, server_s)$} captures multiple dimensions of task-server affinity:
\begin{equation}
\begin{aligned}
&\mathbf{Score} (task_i, server_s)= w_1 \cdot Comp_{hw}(i,s) \\&+ w_2 \cdot Comp_{load}(i,s) + w_3 \cdot Comp_{locality}(i,s) 
\end{aligned}
\end{equation}

The hardware compatibility term $Comp_{hw}(i,s)$ measures how well the task's computational and memory requirements match the server's capabilities. For compute-intensive tasks requiring high throughput, this term favors servers with powerful GPUs like A100 or H100. For memory-intensive tasks involving large model inference, it prioritizes servers with sufficient GPU memory. Formally:
\begin{small}
\begin{equation}
\label{eq: requre}
\begin{aligned}
Comp_{hw}(i,s) =& \min\left(1, \frac{GPU_{compute}^{(s)}}{task_{compute}^{(i)}}\right) \cdot \min\left(1, \frac{GPU_{memory}^{(s)}}{task_{memory}^{(i)}}\right) \\&  \cdot Type_{match}(i,s)
\end{aligned}
\end{equation}
\end{small}
where $Type_{match}(i,s) \in \{0.5, 1.0\}$ indicates whether the server's GPU architecture is optimal for the task type (e.g., A100 for training workloads, RTX series for inference tasks).

The load compatibility term $Comp_{load}(i,s)$ considers the server's current utilization and queue length to promote load balancing:
\begin{small}
\begin{equation}
\label{eq: load}
\begin{aligned}
Comp_{load}(i,s) = \exp\left(-\frac{Util_{current}^{(s)} + Queue_{length}^{(s)}}{Capacity^{(s)}}\right)
\end{aligned}
\end{equation}
\end{small}

This exponential decay function heavily penalizes overloaded servers while providing reasonable scores for moderately loaded ones.

The locality term $Comp_{locality}(i,s)$ accounts for data locality and network proximity. Tasks that require access to cached model weights or intermediate results benefit from being assigned to servers that have recently processed similar tasks:
\begin{small}
\begin{equation}
\label{eq: similar}
\begin{aligned}
Comp_{\text{locality}}(i,s) = \sum_{j \in \text{Recent}_{\text{tasks}}^{(s)}} \frac{\text{Similarity}(\text{task}_i, \text{task}_j)}{\exp(\lambda \cdot (t - \text{timestamp}_j))}
\end{aligned}
\end{equation}
\end{small}
where $Recent_{tasks}^{(s)}$ represents tasks recently processed by server $s$, and $Similarity(\text{task}_i, \text{task}_j) = w_m \cdot \mathbb{I}[\text{model}_i = \text{model}_j] + w_c \cdot \cos(\text{embed}_i, \text{embed}_j)$ measures task similarity based on model type matching and input embedding cosine similarity, with weights $w_m, w_c$.

The greedy assignment algorithm processes tasks in order of urgency (deadline proximity) and computational requirements (resource-intensive tasks first). For each task $i$, it computes the compatibility scores with all available servers and assigns the task to the server with the highest score, provided that the server has sufficient remaining capacity. The algorithm maintains running estimates of server utilization and queue lengths, updating these estimates after each assignment to ensure that subsequent decisions reflect the current allocation state.

To handle situations where no server has sufficient capacity for a task, the algorithm implements a task buffering mechanism that temporarily queues the task until capacity becomes available. The buffering strategy prioritizes tasks based on deadline urgency and can trigger additional server activations if the buffer reaches capacity thresholds.

\begin{algorithm}[t]
    \small
    \setstretch{1}
    \SetKwInOut{Input}{Input}
    \SetKwInOut{Output}{Output}
    \caption{TORTA LOGIC}
    \label{alg: torta}
    \Input{System state $s_t$, Tasks $\mathcal{N}_t$, Cost matrix $C$}
    \Output{Task-server allocation $X_t$, System state $s_{t+1}$}
    \BlankLine
    \tcp{\color{blue}\small{\textbf{Phase 1: Macro-Level} Allocation}}
    \tcp{\color{blue}\small{Demand and Supply Normalization}}
    $\mu_t, \nu_t \leftarrow$ Normalize request and resource distributions\;
        \BlankLine
    \tcp{\color{blue}\small{Optimal Transport}}
    Solve OT problem: $P^* \leftarrow \arg\min_{P \geq 0} \sum_{i,j} C_{i,j} P_{i,j}$\;
    $Prob_{i \rightarrow j} \leftarrow$ Row-normalize $P^*$\;
        \BlankLine
    \tcp{\color{blue}\small{Reinforcement Learning Optimization}}
    $F_t \leftarrow$ Predictor$(U_{t-K:t}, Q_{t-K:t}, H_{t-K:t})$\;
    $A_t \leftarrow \pi_\theta(s_t, F_t, A_{t-1})$\;
        \BlankLine
    \tcp{\color{blue}\small{Regional Task Distribution}}
    \For{each task $i \in \mathcal{N}_t$}{
        $region_i \leftarrow$ Sample from $A_t[origin(i), :]$\;
        Add task $i$ to regional queue $\mathcal{T}_{region_i}$\;
    }
    \BlankLine
    \tcp{\color{blue}\small{\textbf{Phase 2: Micro-Level} Server Selection}}
    \BlankLine
    \For{each region $r = 1$ \textbf{to} $R$}{
        \tcp{\color{blue}\small{Dynamic Server Activation}}
        $N_{target}^{(r,t)} \leftarrow$ Compute target active servers\;
        Update active server set $\mathcal{S}_r^{active}$ based on $N_{target}^{(r,t)}$\;
        \BlankLine
        \tcp{\color{blue}\small{Task-Server Matching}}
        $\mathcal{T}_r \leftarrow$ Sort regional tasks by urgency (deadline proximity)\;
        
        \For{each task $i \in \mathcal{T}_r$}{
            \For{each active server $s \in \mathcal{S}_r^{active}$}{
                $Score(i,s) \leftarrow w_1 \cdot Comp_{hw}(i,s) + w_2 \cdot Comp_{load}(i,s) + w_3 \cdot Comp_{locality}(i,s)$\;
            }
            $s^* \leftarrow \arg\max_s Score(i,s)$\;
            $X_{i,r,s^*}^t \leftarrow 1$\;
            $Util_{s^*} \leftarrow Util_{s^*} + ResourceReq_i$\;
            $Queue_{s^*} \leftarrow Queue_{s^*} + 1$\;
        }
    } 
    \BlankLine
    \tcp{\color{blue}\small{System State Update}}
    Update utilization matrix $U_{t+1}$, queue lengths $Q_{t+1}$\;
    $A_{t-1} \leftarrow A_t$
    \Return{$X_t$, $s_{t+1}$}
\end{algorithm}
\section{Experiments}
\label{sec: exp}
We conduct comprehensive experiments to evaluate the performance of our proposed TORTA scheme under realistic distributed GPU inference scenarios. The experimental configuration is designed to simulate real-world conditions while maintaining controlled variables for fair comparison across different methods.
\subsection{Experimental Setup}

\textit{\textbf{System Configuration.}}
The simulation runs for a total duration of 6 hours, divided into 480 discrete time slots with each time slot lasting 45 seconds. This time granularity provides sufficient resolution to capture dynamic load patterns while maintaining computational tractability for the optimization algorithms. Each GPU inference task has a processing time requirement that follows a uniform distribution, reflecting the heterogeneous nature of real-world inference workloads ranging from lightweight image classification to complex natural language processing tasks.

\textit{\textbf{Network Topology and Infrastructure.}}
We evaluate our scheme across four distinct network topologies that represent different scales and geographical distributions of GPU infrastructure. The topologies range from small-scale regional deployments to large-scale global distributions (Table~\subfigref{tab:infrastructure}{a}), enabling us to assess scalability and performance across various infrastructure configurations.

Each server cluster contains heterogeneous GPU resources with different capabilities and capacities (see Table~\subfigref{tab:infrastructure}{b}). The infrastructure setup is designed to reflect real-world deployment scenarios where different GPU types are optimized for specific workload characteristics.

\textit{\textbf{Cost Model.}}
To provide realistic evaluation conditions, we incorporate real-world electricity costs based on global electricity pricing data. The power consumption and associated costs~\cite{worldpopulationreview2025electricity} vary significantly across geographical regions, reflecting actual operational expenses that cloud service providers face. This cost model enables us to evaluate not only performance metrics but also the economic efficiency of different allocation strategies under realistic pricing constraints.

\textit{\textbf{Baselines.}} To comprehensively evaluate the effectiveness of our proposed scheme, we compare it against several representative baseline methods.
\textbf{\large \texttt{SkyLB:}} This method deploys local load balancers in each region that prioritize local request processing to minimize inter-region communication overhead. When a region reaches full capacity, requests are forwarded to load balancers in other regions with available resources. The system employs a prefix tree structure to route requests from the same user to fixed replicas whenever possible, maintaining session affinity and leveraging cache locality effects~\cite{xia2025skylb}. 
\texttt{\large \textbf{SDIB}} \textit{(Standard Deviation and Idle-time Balanced)}: This approach optimizes for two primary objectives: minimizing the standard deviation of load distribution across servers and reducing the mean idle time of GPU resources. The implementation methodology follows the principles described in MERL-LB~\cite{MERL_Yang_2024}, adapting the multi-objective optimization framework to balance load variance and resource utilization efficiency.
\texttt{\large \textbf{RR}} \textit{(Round-Robin)}: We implement a round-robin scheduling algorithm as a fundamental baseline that provides a simple yet fair allocation strategy. This method applies round-robin selection for both regional GPU resource allocation and intra-cluster server assignment while maintaining necessary capacity and compatibility constraints. The round-robin approach serves as a performance lower bound and demonstrates the benefits of more sophisticated optimization strategies.

It is important to note that the original system scenarios addressed by SkyLB and MERL-LB differ from our distributed GPU inference setting. We have carefully adapted and implemented these algorithms to operate within our simulation environment, preserving their core algorithmic principles while adjusting them to handle the specific characteristics of GPU workload allocation, temporal dependencies, and heterogeneous hardware configurations present in our problem domain.

\subsection{Evaluation Metrics}

We employ three primary metrics to comprehensively assess the performance of our hierarchical allocation scheme, corresponding directly to the three components of our optimization objective:

\textbf{Response Time:} The average response time measures the end-to-end latency experienced by each inference task, calculated from the moment a request is submitted until the processed result is returned to the user. This metric encompasses three components: (1) network delay time for request routing and result transmission, (2) waiting time in queue before processing begins, and (3) actual inference execution time on the assigned GPU server. Response time directly reflects user experience quality and serves as the primary performance indicator for latency-sensitive GPU inference applications.

\textbf{Load Balancing:} We quantify load distribution across the system using inverse of the coefficient of variation $(1/\text{CV})$ of server utilization rates. This metric captures how evenly computational workload is distributed across available GPU resources, with larger values indicating better load balancing. Effective load balancing prevents resource hotspots, ensures efficient utilization of the entire infrastructure, and contributes to overall system stability and predictable performance.

\textbf{Total Cost:} The aggregated cost encompassing both operational expenses and switching overhead incurred during the experimental period. This includes electricity costs calculated based on regional electricity pricing, as well as switching costs measured as allocation differences between consecutive timeslots. This comprehensive cost metric evaluates the economic efficiency of different allocation strategies and their ability to leverage both geographical cost differences and temporal smoothness for overall optimization.
\begin{table}[h]
    \fontsize{8pt}{8pt}\selectfont
    \centering
    \caption{Experimental Infrastructure Configuration}
    \begin{subtable}{0.47\linewidth}  
        \centering
        \renewcommand{\arraystretch}{1.3}
        {\scriptsize (a) Topology Characteristics~\cite{orlowski2010sndlib} \par}
        \begin{tabular}{@{}p{0.18\linewidth}p{0.10\linewidth}p{0.20\linewidth}p{0.20\linewidth}@{}}
        \toprule
        \textbf{Topo.} & \textbf{Node} & \textbf{\makecell{B/W.\\(Gbps)}} & \textbf{\makecell{Lat.\\ (ms)}} \\
        \midrule
        Abilene & 12  & 10 & 25\\
        Polska & 12 & 10 & 45\\
        Gabriel & 25 & 15 & 80\\
        Cost2 & 32 & 20 & 150\\
        \bottomrule
        \end{tabular}
    \end{subtable}
    \hspace{1pt}
    \begin{subtable}{0.5\linewidth}  
        \centering
        \renewcommand{\arraystretch}{1.2}
        {\scriptsize (b) GPU Types and Task Categories \par}
        \begin{tabular}{@{}p{0.13\linewidth}p{0.15\linewidth}p{0.45\linewidth}@{}}
        \toprule
        \textbf{GPU} & \textbf{Count} & \textbf{Task Type} \\
        \midrule
        A100 & 40-60 & Compute-Int. \\
        H100 & 20-40 & Compute-Int. \\
        4090 & 40-60 & Lightweight \\
        V100 & 60-80 & Memory-Int. \\
        T4 & 40-60 & Lightweight \\
        \bottomrule
        \end{tabular}
    \end{subtable}
    \label{tab:infrastructure}
\end{table}

\subsection{Experimental Analysis}
\subsubsection{Response Time Performance Analysis}
\begin{figure*}[t]
\centering
\includegraphics[width=\textwidth]{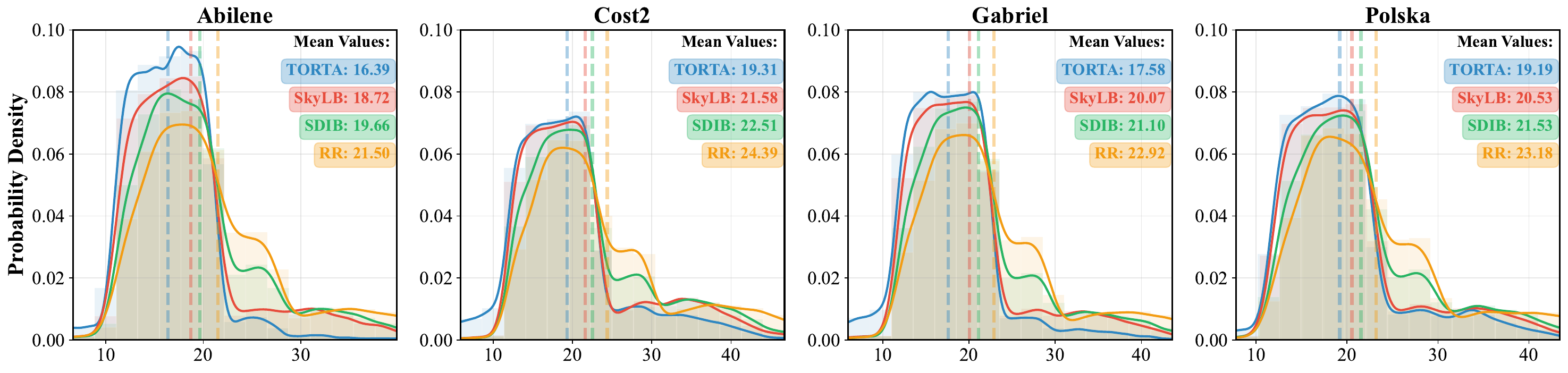}
\caption{Response time probability distributions across different network topologies. The shaded regions represent the core distribution range, with vertical dashed lines indicating mean values for each algorithm.}
\label{fig:response_time}
\end{figure*}
\begin{figure*}[t]
\centering
\includegraphics[width=\textwidth]{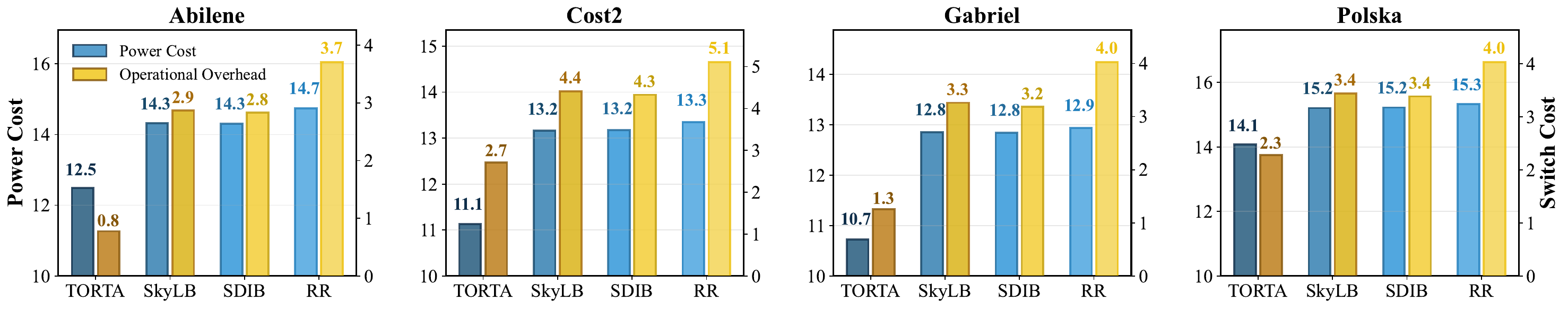}
\caption{Power cost and operational overhead comparison across different network topologies. Power costs are measured in thousands of dollars (\$K), while operational overhead represent the relative overhead (normalized units).}
\label{fig:cost_comparison}
\end{figure*}
The response time distributions shown in Figure~\ref{fig:response_time} demonstrate the effectiveness of our TORTA scheme across four different network topologies. TORTA consistently achieves the fastest average task completion times, with mean response times of 16.39s, 19.31s, 17.58s, and 19.19s for Abilene, Cost2, Gabriel, and Polska topologies respectively. This represents substantial improvements of 12.4\% (vs. SkyLB: 18.72s), 10.5\% (vs. SkyLB: 21.58s), 12.4\% (vs. SkyLB: 20.07s), and 6.5\% (vs. SkyLB: 20.53s) compared to the best-performing baseline method across each topology.

A detailed examination of the distribution tails reveals that TORTA performs particularly well for high-latency tasks, effectively reducing task queuing times through intelligent scheduling decisions. This improvement is especially pronounced in the right tail of the distributions, where TORTA shows substantially lower probability density for long response times compared to other methods. The temporal-aware optimization enables the system to proactively manage resource allocation, preventing the accumulation of tasks in queues that typically leads to extended completion times.

The performance gap between TORTA and baseline methods is notably smaller on the Polska topology compared to other network configurations. This phenomenon can be attributed to the superior connectivity characteristics of the Polska topology, which provides more routing options and reduces the impact of suboptimal scheduling decisions. When the underlying network infrastructure offers better connectivity, the advantages of sophisticated temporal optimization become less pronounced, as even simpler algorithms can achieve reasonable performance through the network's inherent redundancy.

\subsubsection{Cost Optimization and Economic Efficiency} 
Figure \ref{fig:cost_comparison} presents a comprehensive cost analysis comparing TORTA with baseline methods across four network topologies. TORTA consistently achieves the lowest power costs across all topologies: 12.5K (vs. SkyLB's 14.3K, 12.6\% reduction), 11.1K (vs. SkyLB's 13.2K, 15.9\% reduction), 10.7K (vs. SkyLB's 12.8K, 16.4\% reduction), and 14.1K (vs. SkyLB's 15.2K, 7.2\% reduction) for Abilene, Cost2, Gabriel, and Polska topologies respectively. Compared to other baselines, TORTA achieves similar improvements of 12-16\% over SDIB and RR methods.

The optimal transport component enables TORTA to leverage geographical cost variations by routing tasks to regions with lower electricity prices while maintaining network latency constraints. This cost-aware allocation strategy is particularly effective in Cost2 and Gabriel topologies, where TORTA achieves the most substantial power cost reductions of 15.9\% and 16.4\% respectively.

For operational overhead, TORTA demonstrates dramatic improvements: 0.8 vs. 2.9 (72\% reduction), 2.7 vs. 4.4 (39\% reduction), 1.3 vs. 3.3 (61\% reduction), and 2.3 vs. 3.4 (32\% reduction) compared to SkyLB across the four topologies. Compared to other baselines, TORTA achieves 71-79\% reductions vs. SDIB and 73-78\% reductions vs. RR methods. The operational overhead metric represents normalized system transition costs (measured in planning units) that encompasses actual model loading, memory management, and server state changes, extending beyond the theoretical switching cost used in our analysis (which measures allocation matrix differences as a proxy). TORTA's micro-level server selection employs proactive preheating and locality-aware assignment, while SkyLB and SDIB use cache-aware strategies and RR lacks such optimizations, explaining the varied overhead levels across methods.

\begin{figure*}[t]
\centering
\includegraphics[width=\textwidth]{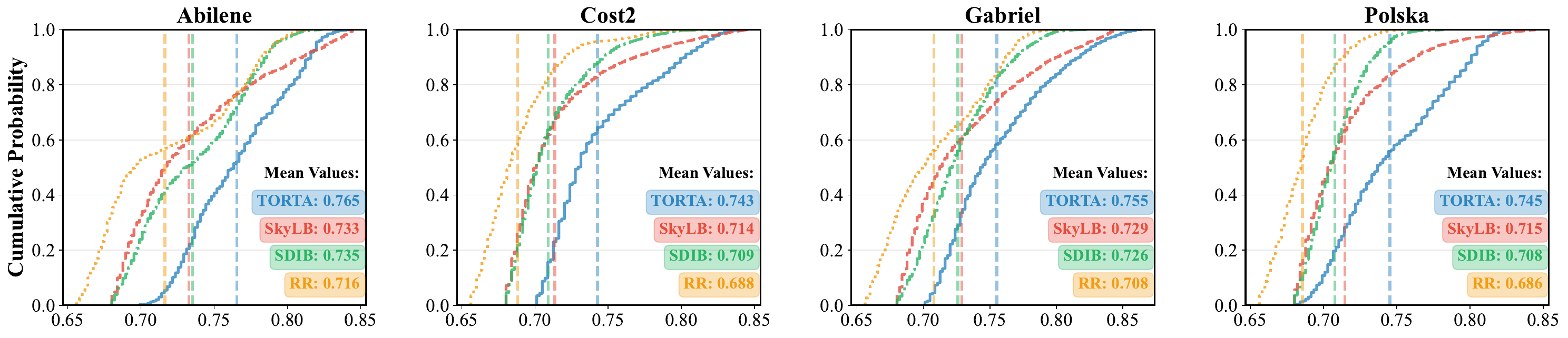}
\caption{Cumulative distribution function (CDF) of load balance coefficients across different network topologies. The load balance coefficient is calculated as the inverse of one plus the coefficient of variation of server utilization rates. Higher values indicate better load distribution. Vertical dashed lines represent mean values for each algorithm.}
\label{fig:load_balance_cdf}
\vspace{-0.1cm}
\end{figure*}
\subsubsection{Load Balance} Figure \ref{fig:load_balance_cdf} demonstrates the superior load balancing performance of TORTA across all tested network topologies. TORTA consistently achieves the highest load balance coefficients with mean values of 0.765, 0.743, 0.755, and 0.745 for Abilene, Cost2, Gabriel, and Polska topologies respectively. Compared to the best-performing baseline SkyLB, this represents improvements of 4.4\% (0.765 vs 0.733), 4.1\% (0.743 vs 0.714), 3.6\% (0.755 vs 0.729), and 4.2\% (0.745 vs 0.715) across the four topologies. The improvements over other baselines are more substantial: 4.1-6.9\% vs SDIB and 6.8-11.5\% vs RR methods. The load balance coefficient $LB$ is defined as:
\begin{equation}
    \begin{aligned}
    LB = \frac{1}{1 + CV_{util}}  = \frac{1}{1 + \frac{\sqrt{\frac{1}{N}\sum_{i=1}^{N}(U_i - \bar{U})^2}}{\frac{1}{N}\sum_{i=1}^{N}U_i}}
    \end{aligned}
\end{equation}
where $U_i$ represents the utilization rate of server $i$, $\bar{U}$ is the mean utilization across all servers, $N$ is the total number of active servers, and $CV_{util}$ is the coefficient of variation. These results represent substantial improvements over baseline methods, indicating more uniform resource utilization across the distributed infrastructure.

The enhanced load balancing performance can be attributed to TORTA's smooth allocation strategy at the macro level, which prevents dramatic load variations between consecutive time slots within server clusters. The temporal smoothness term in our reinforcement learning objective function explicitly penalizes large allocation changes, ensuring gradual transitions that avoid sudden spikes in server workloads. This temporal consistency is particularly important for maintaining stable system performance and preventing cascading effects that can lead to load imbalances.

At the micro level, TORTA employs dynamic server activation through Equation~\ref{eq: activate} and load-aware server selection via Equation~\ref{eq: load} to maintain balanced utilization across active servers. The exponential decay function in the load compatibility term heavily penalizes assignments to already overloaded servers, naturally distributing tasks toward underutilized resources. Additionally, the dynamic activation strategy ensures that the number of active servers adapts to current demand, preventing situations where a small number of servers handle disproportionate workloads while others remain idle.

\subsubsection{Details about Performance Metrics}
\begin{figure}[t]
\centering
\includegraphics[width=\linewidth]{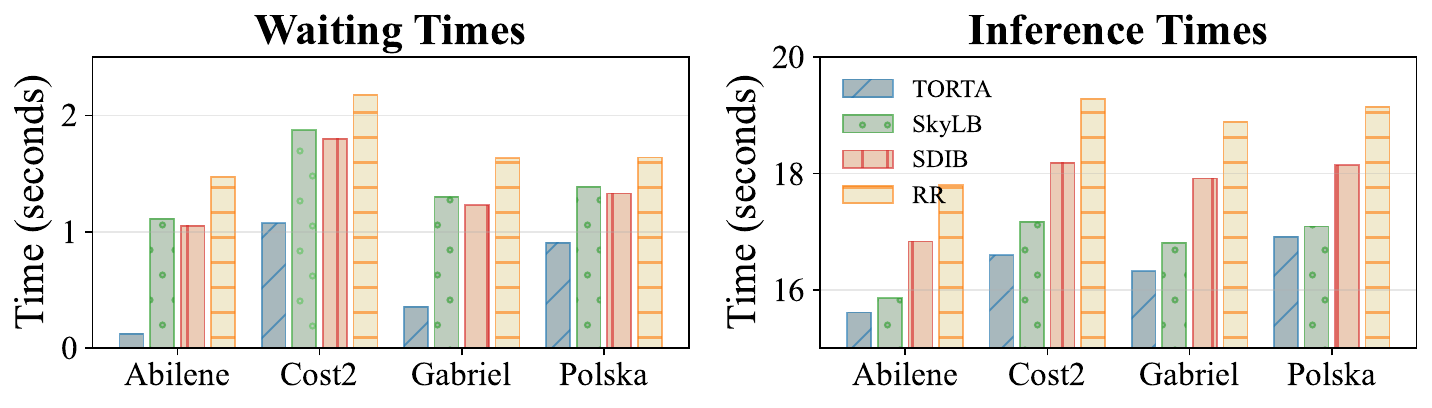}
\caption{Breakdown of response time components across different network topologies.}
\label{fig:response_time_breakdown}
\end{figure}

Figure \ref{fig:response_time_breakdown} provides a detailed breakdown of response time components, showing that TORTA achieves superior performance in both waiting times and inference times across all network topologies. The waiting times for TORTA range from 0.3 to 1.1 seconds compared to 1.2-2.4 seconds for baseline methods, representing reductions of approximately 50-75\%. This improvement stems from TORTA's proactive resource management and temporal-aware allocation strategy that prevents task accumulation.

Figure \ref{fig:prediction_accuracy_impact} illustrates the relationship between demand prediction accuracy and system performance for TORTA compared to baseline methods. Since the baseline methods (RR, SkyLB, and SDIB) do not incorporate prediction modules, their performance remains constant across different accuracy levels, as shown by the horizontal lines. In contrast, TORTA demonstrates clear performance improvements with increased prediction accuracy, with response times decreasing from approximately 20.5 seconds at 0.1 accuracy to 17.5 seconds at 0.9 accuracy. The shaded region around TORTA's curve represents the confidence interval derived from multiple experimental runs.

A critical observation is that TORTA begins to outperform all baseline methods when prediction accuracy reaches 0.5, which represents a relatively achievable forecasting requirement. The prediction accuracy $PA$ is defined as:
\begin{equation}
    \begin{aligned}
        PA = \exp\left(-\frac{1}{T}\sum_{t=1}^{T}\frac{|F_t^{pred} - F_t^{actual}|}{F_t^{actual} + \epsilon}\right)
    \end{aligned}
    \label{eq: pa}
\end{equation}
where $F_t^{pred}$ represents the predicted traffic volume at time $t$, $F_t^{actual}$ is the actual observed traffic, $T$ is the total number of time slots, and $\epsilon$ is a small constant to prevent division by zero. This metric ranges from 0 to 1, with higher values indicating more accurate predictions.
\begin{figure}[t]
\centering
\includegraphics[width=\linewidth]{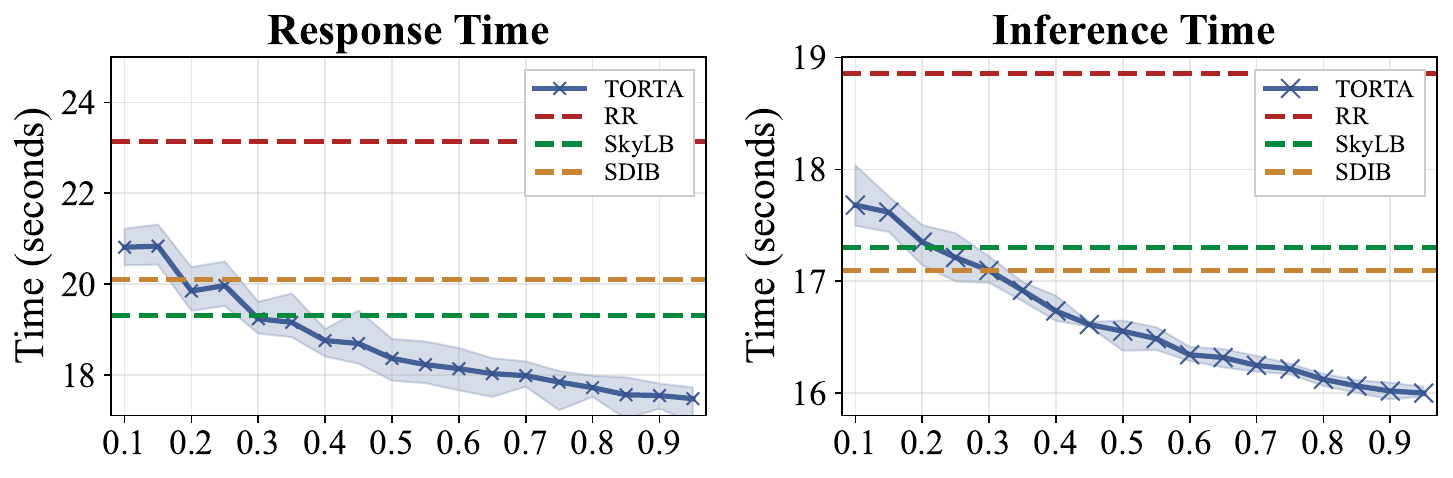}
\caption{Impact of demand prediction accuracy on system performance. Left panel shows response times and right panel shows inference times as functions of prediction accuracy. The prediction accuracy is defined as Equation~\ref{eq: pa}.}
\label{fig:prediction_accuracy_impact}
\vspace{-0.5cm}
\end{figure}

This threshold (0.4) suggests that the prediction module does not need to achieve exceptional accuracy to provide meaningful benefits. The 0.4 accuracy level corresponds to reasonable demand forecasting capabilities that can be obtained using standard time series prediction techniques, making TORTA practical for real-world deployment scenarios.

Even when prediction accuracy falls below the 0.4 threshold, TORTA's performance degradation is gradual rather than catastrophic, indicating the robustness of the two-stage allocation scheme. The micro-level server selection component effectively compensates for macro-level prediction inaccuracies through real-time adaptation and intelligent task-server matching. This resilience ensures that TORTA remains competitive even when demand forecasting is challenging, while providing substantial performance gains when reasonable prediction accuracy is available.

\section{Related Work}
\label{sec: related_works}
The optimization of large language model inference has primarily focused on addressing the fundamental mismatch between the compute-intensive prefill phase, which processes all input tokens in parallel, and the subsequent memory-intensive decode phase, which generates tokens sequentially~\cite{donisch2024inferenceopt-llm}. Several approaches have emerged to tackle this challenge through various architectural and algorithmic innovations. Methods such as speculative decoding \cite{xia2025speculativedecoding, hu2025speculativedecodingbeyondindepth, marzollo2024sssdsimplyscalablespeculativedecoding} attempt to predict multiple tokens simultaneously to reduce the sequential bottleneck, while attention optimization techniques~\cite{pan2024instinferinstorageattentionoffloading, ribar2024sparqattentionbandwidthefficientllm, ye2025flashinferefficientcustomizableattention} focus on improving memory efficiency during the attention computation. Model compression approaches~\cite{xie2025reimaginingcompresssionmemoryaccessllm, hansenpalmus2024communicationcompression, xu2023compresspromptimprovingaccuracyefficiency} reduce computational requirements through quantization and pruning, though they may compromise model accuracy. Some works target Mixture-of-Experts (MoE) models, aiming to achieve expert load balancing~\cite{hwang2024pre, hwang2023tutel} and faster inter-expert communication~\cite{deepspeed2022, li2024optimizing}. These acceleration techniques primarily operate at the algorithmic level within individual inference engines, complementing but not replacing the need for effective resource allocation across distributed systems.

LLM inference scheduling encompasses both local scheduling within individual servers and distributed scheduling across clusters. At the engine level, traditional approaches have evolved from simple first-come-first-served strategies~\cite{Xie2025NVIDIA-fcfs} to sophisticated predictive methods. Non-preemptive scheduling methods \cite{yu2022non-preempitve} pre-allocate maximum possible memory based on context length limits, often leading to significant memory underutilization when requests complete early. Predictive scheduling approaches~\cite{zheng2023-pia, shahout2024dontstopnowembedding-trail, fu2024efficientllmschedulinglearning-ltr, choi2025elisefficientllmiterative-elis}  leverage various prediction mechanisms to estimate response lengths and optimize task ordering, though these methods differ fundamentally from our temporal demand forecasting as they focus on individual request characteristics rather than system-wide load patterns over time. For instance, systems like PiA~\cite{zheng2023-pia} use the LLM itself for length estimation, while TRAIL~\cite{shahout2024dontstopnowembedding-trail} implements shortest predicted remaining processing time strategies. Dynamic scheduling frameworks~\cite{lin2025bulletboostinggpuutilization-bullet, kim2025optimizingllminferencedatabase-infermax, wu2024fastdistributedinferenceserving-fastserve, kossmann2025gpuhalfemptyhalffullpractical-larry} emphasize real-time adaptation and responsive resource management, with systems like Bullet~\cite{lin2025bulletboostinggpuutilization-bullet} addressing prefill/decode mismatches through concurrent execution and InferMax \cite{kim2025optimizingllminferencedatabase-infermax} employing cost models for batch time prediction. Systems like PETALS \cite{borzunov2023-petals, Choudhury2024mast, dong2025singleaiclustersurvey-geodist} address fault-tolerant inference across distributed commodity hardware, and provide global scheduling for large-scale training across data centers. However, these approaches remain fundamentally reactive, optimizing within individual time windows without considering temporal dependencies across scheduling decisions.

\section{Conclusion}
\label{sec: conclusion}
This paper presents TORTA, a novel scheme for distributed GPU inference optimization that addresses the fundamental limitations of existing reactive scheduling approaches. By integrating reinforcement learning with optimal transport theory and incorporating temporal dependencies into the allocation process, TORTA achieves improvements in response times, load balancing, and operational overhead across diverse network topologies. The two-layer hierarchical architecture effectively decomposes the complex spatiotemporal optimization problem, enabling scalable solutions that consider both geographical resource distribution and temporal dynamics.

Experimental results demonstrate that TORTA consistently outperforms baseline methods across multiple metrics, achieving response time reductions of up to 15\% and substantial improvements in load balancing coefficients. The scheme exhibits robustness to prediction inaccuracies and maintains competitive performance even when demand forecasting accuracy is modest. However, several limitations remain: TORTA's performance gains depend on prediction quality with diminishing returns below 40\% accuracy, and the two-layer optimization introduces computational overhead for extremely large-scale deployments. Future work will focus on extending the framework to multi-tenant scenarios, investigating adaptive prediction mechanisms for improved forecasting accuracy, and developing lightweight variants for resource-constrained environments.

\appendices
\section{Theoretical Analysis of TARTO}
\label{prof: torto}
In this appendix, we provide a rigorous theoretical analysis demonstrating that TORTA can provably outperform the performance upper bound of single-timeslot reactive methods in multi-timeslot scenarios. Our proof consists of three main components: (1) establishing the performance upper bound for single-timeslot methods, (2) analyzing the switching cost convergence properties, and (3) deriving conditions under which TORTA surpasses this upper bound.

\subsection{ Upper Bound for Single-Timeslot Methods}
\label{prof: upper_bound}
We begin by establishing the theoretical upper bound for any single-timeslot optimization method when applied across multiple timeslots.

\begin{Definition}[Single-Timeslot Method]
A single-timeslot method $\mathcal{M}$ is a memoryless allocation strategy where the decision at timeslot $t$ depends only on the current state $(s_t, \mu_t, \nu_t)$, formally: $A_t^{\mathcal{M}} = f_{\mathcal{M}}(\mu_t, \nu_t, C_t)$.
\end{Definition}

\begin{Theorem}[Optimal Transport Achieves Single-Timeslot Upper Bound]
\label{thm:ot_upper_bound}
For any single timeslot $t$, the optimal transport solution $P^*$ minimizes both response time and power cost simultaneously, establishing the performance upper bound for single-timeslot optimization.
\end{Theorem}

\begin{Proof}
Let $w_j = \sum_{i=1}^{R} P^*_{i,j}$ denote the total workload assigned to region $j$, and $T_j = w_j/c_j$ be the response time in region $j$ with capacity $c_j$. The maximum response time is $T_{\max} = \max_{j \in \mathcal{R}} T_j$.

Suppose there exists an allocation $P'$ such that $T'_{\max} < T_{\max}$ but $P'$ is not the optimal transport solution. Since $P'$ is suboptimal for OT, there exist regions $i, j$ such that $w'_i/c_i \neq w'_j/c_j$.

Without loss of generality, assume $w'_i/c_i > w'_j/c_j$. We can construct a new allocation $\tilde{P}$ by redistributing tasks from region $i$ to region $j$ such that $\tilde{w}_i/c_i = \tilde{w}_j/c_j$. This redistribution satisfies:
\begin{align}
\sum_j \tilde{w}_j &= \sum_j w'_j \quad \text{(workload conservation)} \\
\sum_j \tilde{P}_{i,j} &= \mu_t^{(i)} \quad \text{(demand constraint)}
\end{align}

The new maximum response time satisfies:
$$\tilde{T}_{\max} = \max\left\{\frac{\tilde{w}_i}{c_i}, \frac{\tilde{w}_j}{c_j}\right\} < \max\left\{\frac{w'_i}{c_i}, \frac{w'_j}{c_j}\right\} = T'_{\max}$$

This contradicts the assumption that $P^*$ is optimal. Therefore, OT achieves the minimum possible response time.

For power cost optimality, since OT minimizes $\sum_{i,j} C_{i,j} P_{i,j}$ where $C_{i,j}$ incorporates regional electricity price differences, it naturally achieves minimum power cost. $\hfill\blacksquare$
\end{Proof}

\subsection{Switching Cost Convergence Analysis}

Next, we analyze the switching costs incurred by single-timeslot methods across multiple timeslots.

\begin{assumption}[Temporal Independence]
\label{ass:temporal_independence}
The resource distributions $\{\nu_t\}_{t=1}^T$ and request distributions $\{\mu_t\}_{t=1}^T$ are temporally independent: $\nu_t \perp \nu_{t'}$ and $\mu_t \perp \mu_{t'}$ for all $t \neq t'$.
\end{assumption}

\begin{Theorem}[Switching Cost Convergence]
\label{thm:switching_convergence}
Under Assumption \ref{ass:temporal_independence}, for any single-timeslot method $\mathcal{M}$, the expected switching cost between consecutive timeslots converges to a method-independent constant $K_0$.
\end{Theorem}

\begin{Proof}
Let $\Delta_t^{\mathcal{M}} = ||A_t^{\mathcal{M}} - A_{t-1}^{\mathcal{M}}||_F^2$ denote the switching cost. We have:
\begin{align}
\mathbb{E}[\Delta_t^{\mathcal{M}}] &= \mathbb{E}[||A_t^{\mathcal{M}} - A_{t-1}^{\mathcal{M}}||_F^2] \\
&= \mathbb{E}[||A_t^{\mathcal{M}}||_F^2] + \mathbb{E}[||A_{t-1}^{\mathcal{M}}||_F^2] - 2\mathbb{E}[\langle A_t^{\mathcal{M}}, A_{t-1}^{\mathcal{M}} \rangle]
\end{align}

By temporal independence and the memoryless property of $\mathcal{M}$:
$$\mathbb{E}[\langle A_t^{\mathcal{M}}, A_{t-1}^{\mathcal{M}} \rangle] = \langle \mathbb{E}[A_t^{\mathcal{M}}], \mathbb{E}[A_{t-1}^{\mathcal{M}}] \rangle = ||\mathbb{E}[A^{\mathcal{M}}]||_F^2$$

Therefore:
\begin{align}
\mathbb{E}[\Delta_t^{\mathcal{M}}] &= 2\mathbb{E}[||A^{\mathcal{M}}||_F^2] - 2||\mathbb{E}[A^{\mathcal{M}}]||_F^2 \\
&= 2\text{Var}(A^{\mathcal{M}}) = K_0
\end{align}

The key insight is that for methods satisfying the same constraints (demand and capacity), the variance is primarily determined by input distribution randomness rather than the optimization algorithm itself, leading to method-independent convergence.$\hfill\blacksquare$
\end{Proof}

\begin{Corollary}[Performance Upper Bound for Single-Timeslot Methods]
\label{cor:performance_upper_bound}
For any single-timeslot method $\mathcal{M}$ applied over $T$ timeslots, the total expected cost satisfies:
$$\mathbb{E}[\text{Cost}_{\mathcal{M}}^{\text{total}}] \geq \sum_{t=1}^T (\text{RT}_t^{\text{OT}} + \beta \cdot \text{PC}_t^{\text{OT}}) + \alpha \cdot K_0 (T-1)$$
where $\text{RT}_t^{\text{OT}}$ and $\text{PC}_t^{\text{OT}}$ are the optimal response time and power cost achieved by OT at timeslot $t$.
\end{Corollary}

\subsection{TORTA's Performance Guarantee}
Now we establish conditions under which TORTA can surpass the performance upper bound of single-timeslot methods.
\begin{Definition}[TORTA Training Objective]
TORTA is trained with a constrained loss function:
$$\mathcal{L}_{\text{total}} = \mathcal{L}_{\text{PPO}} + \gamma \cdot \mathcal{L}_{\epsilon}(s, \epsilon) + \delta \cdot \mathcal{L}_{s}(s, \epsilon)$$
where:
\begin{align}
\mathcal{L}_{\epsilon}(s, \epsilon) &= \max\left(0, \frac{||B_t||_F - \epsilon_{\max}(s)}{\epsilon_0}\right) \\
\mathcal{L}_{s}(s, \epsilon) &= \max\left(0, \frac{s_{\min}(\epsilon) - s}{s_0}\right)
\end{align}
with $B_t = A_t^{\text{RL}} - A_t^{\text{OT}}$ and designed to ensure the performance advantage condition.
\end{Definition}

\begin{Theorem}[TORTA Performance Guarantee]
\label{thm:torta_guarantee}
When trained with the constrained loss function, TORTA satisfies:
\begin{enumerate}
\item \textbf{Switching cost improvement}: $\mathbb{E}[\Delta_t^{\text{RL}}] \leq K_0/s$ for some $s > 1$
\item \textbf{Single-timeslot performance control}: $\mathbb{E}[||B_t||_F] \leq \epsilon$
\item \textbf{Performance advantage condition}: When $\frac{1-1/s}{\epsilon} > \frac{L_R + \beta L_P}{\alpha \cdot K_0}$, TORTA strictly outperforms all single-timeslot methods.
\end{enumerate}
\end{Theorem}

\begin{Proof}
\textbf{Part 1 (Switching cost improvement):} The PPO training objective includes an explicit smoothness reward:
$$r_t^{\text{smooth}} = -||A_t^{\text{RL}} - A_{t-1}^{\text{RL}}||_F^2$$

Through policy gradient optimization, the converged policy satisfies:
$$\mathbb{E}[\Delta_t^{\text{RL}}] = \mathbb{E}[||A_t^{\text{RL}} - A_{t-1}^{\text{RL}}||_F^2] \leq \frac{K_0}{s}$$

where $s > 1$ is the switching cost improvement factor controlled by the smoothness weight.

\textbf{Part 2 (Single-timeslot performance control):} Let $B_t = A_t^{\text{RL}} - A_t^{\text{OT}}$. Under the OT supervision constraint, the performance degradation is bounded by:
\begin{align}
\delta R_t &\leq L_R ||B_t||_F \\
\delta P_t &\leq L_P ||B_t||_F
\end{align}
where $L_R, L_P$ are Lipschitz constants. The constraint $\mathbb{E}[||B_t||_F]  \leq \epsilon$ is enforced by $\mathcal{L}_{\epsilon}$.

\textbf{Part 3 (Performance advantage):} The total cost difference between TORTA and the single-timeslot upper bound is:
\begin{align}
\text{Diff} &= \text{Cost}_{\text{TORTA}} - \text{Cost}_{\text{Upper}} \\
&= \sum_{t=1}^T (\delta R_t + \beta \delta P_t) + \alpha \sum_{t=1}^{T-1} (\Delta_t^{\text{RL}} - K_0) \\
&\leq \sum_{t=1}^T [(L_R + \beta L_P)\epsilon] + \alpha \sum_{t=1}^{T} \left(\frac{K_0}{s} - K_0\right) \\
&= T\left[(L_R + \beta L_P)\epsilon - \alpha K_0\left(1 - \frac{1}{s}\right)\right]
\end{align}

When the performance advantage condition $\frac{1-1/s}{\epsilon} > \frac{L_R + \beta L_P}{\alpha \cdot K_0}$ holds, we have:
$$1 - \frac{1}{s} > \epsilon \cdot \frac{L_R + \beta L_P}{\alpha \cdot K_0}$$

Substituting this into the Diff expression:
\begin{align}
\text{Diff} &< T\left[(L_R + \beta L_P)\epsilon - \alpha K_0 \cdot \epsilon \cdot \frac{L_R + \beta L_P}{\alpha \cdot K_0}\right] \\
&= T\left[(L_R + \beta L_P)\epsilon - (L_R + \beta L_P)\epsilon\right] = 0
\end{align}

Therefore, $\text{Cost}_{\text{TORTA}} < \text{Cost}_{\text{Upper}}$, proving TORTA's strict superiority.
$\hfill\blacksquare$
\end{Proof}

\begin{remark}
The performance advantage condition $\frac{1-1/s}{\epsilon} > \frac{L_R + \beta L_P}{\alpha \cdot K_0}$ provides explicit guidance for training TORTA. It shows that superior performance can be achieved by either: (1) increasing the switching cost improvement factor $s$, (2) tightening the OT deviation bound $\epsilon$, or (3) increasing the switching cost weight $\alpha$ relative to response time and power cost weights.
\end{remark}

\section{Training Implementation Details}
\label{detail: training}
\begin{algorithm}[htbp]
    \small
    \setstretch{1}
    \SetKwInOut{Input}{Input}
    \SetKwInOut{Output}{Output}
    \caption{TORTA Training Procedure}
    \label{alg:torta_training}
    \Input{Historical data $\mathcal{D}$, Environment parameters, Target condition $\frac{1-1/s}{\epsilon} > \frac{L_R + \beta L_P}{\alpha \cdot K_0}$}
    \Output{Trained policy $\pi_\theta$, predictor $P_\psi$}
    \BlankLine
    
    \tcp{\color{blue}\small{Initialization and Baseline Parameter Estimation}}
    Initialize policy network $\pi_\theta$, value network $V_\phi$, predictor $P_\psi$\;
    Pre-train predictor $P_\psi$ on historical data $\mathcal{D}$\;
    
    \BlankLine
    \tcp{\color{blue}\small{Estimate Baseline Parameters for Theoretical Conditions}}
    $K_0 \leftarrow$ Compute baseline switching costs from reactive methods\;
    $L_R, L_P \leftarrow$ Estimate Lipschitz constants via finite differences\;
    Set target parameters: $\epsilon_{\text{target}} = 0.15$, $s_{\text{target}} = 2.5$\;
    
    \BlankLine
    \For{epoch $e = 1$ \textbf{to} $E$}{
        \tcp{\color{blue}\small{Collect Experience with Constraint Monitoring}}
        \For{environment step $t = 1$ \textbf{to} $T$}{
            Collect experience $(s_t, a_t, r_t, s_{t+1})$ from environment\;
            Compute OT baseline $A_t^{\text{OT}}$ using current $\mu_t, \nu_t$\;
            Calculate deviation $B_t = A_t^{\text{RL}} - A_t^{\text{OT}}$\;
            Update switching cost estimate $s_{\text{current}} \leftarrow \frac{K_0}{\mathbb{E}[\Delta_t^{\text{RL}}]}$\;
            Update constraint weights $\gamma_t, \delta_t$ based on current violations\;
        }
        
        \BlankLine
        \tcp{\color{blue}\small{Policy Update with Theoretical Constraints}}
        Compute PPO loss $\mathcal{L}_{\text{PPO}}$ from collected experience\;
        Compute constraint losses $\mathcal{L}_{\epsilon} = \max(0, \frac{||B_t||_F - \epsilon_{\text{target}}}{\epsilon_0})$\;
        Compute constraint losses $\mathcal{L}_s = \max(0, \frac{s_{\text{target}} - s_{\text{current}}}{s_0})$\;
        
        \BlankLine
        \tcp{\color{blue}\small{Gradient Update and Condition Validation}}
        $\theta \leftarrow \theta - \nabla_\theta(\mathcal{L}_{\text{PPO}} + \gamma_t \mathcal{L}_{\epsilon} + \delta_t \mathcal{L}_s)$\;
        Validate performance advantage condition: $\frac{1-1/s_{\text{current}}}{\epsilon_{\text{current}}} \stackrel{?}{>} \frac{L_R + \beta L_P}{\alpha \cdot K_0}$\;
        
        \textbf{if} condition violated \textbf{then} increase $\gamma_t, \delta_t$ by factor 1.5\;
    }
    
    \Return{$\pi_\theta$, $P_\psi$}
\end{algorithm}
\subsection{Network Details}
For Reinforcement learning, the policy network $\pi_\theta$ is implemented as a fully connected neural network with three hidden layers (256, 512, 256 units) and ReLU activations. The output layer uses a softmax activation to generate probability distributions over regional allocations. The value network $V_\phi$ shares the same architecture but outputs scalar state values. We use the Adam optimizer with an initial learning rate of $3 \times 10^{-4}$ and exponential decay with factor 0.995 every 100 episodes.

The demand predictor $F_t = \text{Predictor}(U_{t-K:t}, Q_{t-K:t},\\ H_{t-K:t})$ employs a simple yet effective multi-layer perceptron architecture. The input consists of concatenated historical features over $K=5$ timeslots, yielding $\text{Input} = \text{Concat}([U_{t-5:t}, Q_{t-5:t}, H_{t-5:t}]) \in \mathbb{R}^{5 \times 3R}$. The network architecture follows a standard 3-layer design with ReLU activations: input layer ($15R$ dimensions) $\rightarrow$ hidden layer (512 units) $\rightarrow$ hidden layer (256 units) $\rightarrow$ output layer ($R$ dimensions with softmax). The training objective minimizes mean squared error with L2 regularization: $\mathcal{L}_{\text{pred}} = \frac{1}{N} \sum_{i=1}^N ||F_{t+1}^{\text{pred}} - F_{t+1}^{\text{actual}}||_2^2 + \lambda ||\theta||_2^2$ where $\lambda = 10^{-4}$.

\subsection{Constrained Training Framework}
To ensure the performance advantage condition $\frac{1-1/s}{\epsilon} > \frac{L_R + \beta L_P}{\alpha \cdot K_0}$ is satisfied during training, we implement adaptive constraint enforcement. The constraint weights $\gamma_t$ and $\delta_t$ for terms $\mathcal{L}_{\epsilon}$ and $\mathcal{L}_s$ are dynamically adjusted based on current performance: $\gamma_t = \gamma_0 \cdot \exp(\alpha_{\gamma} \cdot \max(0, ||B_t||_F - \epsilon_{\text{target}}))$ and $\delta_t = \delta_0 \cdot \exp(\alpha_{\delta} \cdot \max(0, s_{\text{target}} - s_{\text{current}}))$, where $B_t = A_t^{\text{RL}} - A_t^{\text{OT}}$ represents the deviation from the optimal transport solution.

The critical system parameters required for theoretical guarantees are estimated continuously during training. The baseline switching cost $K_0$ is computed as a moving average of switching costs from reactive baseline methods over  past timeslots. The Lipschitz constants $L_R$ and $L_P$ for response time and power cost functions are estimated empirically using finite differences over small perturbations in allocation matrices. The current switching cost improvement factor $s_{\text{current}}$ is calculated as the ratio of baseline switching costs to current RL switching costs, updated every 10 training episodes.

\bibliographystyle{plain}
\renewcommand{\bibfont}{\footnotesize} 
\bibliography{cas-refs}


\end{document}